\documentclass[sts]{imsart}

\usepackage{amsthm,amsmath, amssymb, color}
\RequirePackage[colorlinks,citecolor=blue,urlcolor=blue]{hyperref}
\usepackage{xcolor,colortbl}
\RequirePackage{hypernat}

\usepackage{color}
\usepackage{amsfonts, fancybox}
\usepackage{amssymb}
\usepackage[american]{babel}
\usepackage{psfrag,color}
\usepackage{dcolumn}
\usepackage{subfigure}
\usepackage{flafter, bm, dsfont}
\usepackage[section]{placeins}

\usepackage{multirow}
\usepackage{color}
\usepackage{amsmath}
\usepackage{float}
\usepackage{mathrsfs}
\usepackage{amsfonts}
\usepackage{dsfont}
\usepackage{amssymb}
\usepackage{algorithmic, algorithm}
\usepackage{wrapfig}

\usepackage{amssymb, latexsym}

\usepackage{graphicx}
\usepackage{epstopdf}

\usepackage{cite}
\usepackage{boxedminipage}
\usepackage{tkz-graph}
\usetikzlibrary{external}
\newcommand{\bc}{\begin{center}}
\newcommand{\ec}{\end{center}}
\newcommand{\ba}{\begin{array}}
\newcommand{\ea}{\end{array}}
\newcommand{\be}{\begin{eqnarray}}
\newcommand{\ee}{\end{eqnarray}}
\newcommand{\bel}{\begin{eqnarray}\label}
\newcommand{\eel}{\end{eqnarray}}
\newcommand{\bes}{\begin{eqnarray*}}
\newcommand{\ees}{\end{eqnarray*}}
\newcommand{\bn}{\begin{enumerate}}
\newcommand{\en}{\end{enumerate}}

\newcommand{\KL}{\mathsf{KL}}

\newcommand{\bern}{\mathsf{Bern}}
\newcommand{\unif}{\mathsf{Unif}}

\newcommand{\aucb}{\textsc{UCBid }}
\newcommand{\expthree}{\textsc{Exp3}}
\newcommand{\exptree}{\textsc{ExpTree}}
\newcommand{\exptreep}{\textsc{ExpTree.P}}

\newcommand{\tinyplus}{\scalebox{0.7}{+}}
\newcommand{\tinyminus}{\scalebox{0.8}{-}}

\definecolor{MIT}{cmyk}{.24, 1.00, .78, .17} 
\definecolor{pink}{cmyk}{0, 1, 0, 0} 
\definecolor{darkgreen}{cmyk}{1,0, 1, 0}

\newcommand{\iid}{{\it i.i.d.\ }}

\newtheorem{lemma}{Lemma}
\newtheorem{definition}{Definition}

\newtheorem{theorem}{Theorem}

\newcommand{\indic}[1]{\mathds{1}\{#1\}}

\newcommand{\simiid}{\overset{\text\small\rm{iid}}{\sim}}

\newcommand{\DS}{\displaystyle}

\newcommand{\cB}{\mathcal{B}}

\newcommand{\cF}{\mathcal{F}}

\newcommand{\cL}{\mathcal{L}}

\newcommand{\cW}{\mathcal{W}}

\newcommand{\R}{{\rm I}\kern-0.18em{\rm R}}
\newcommand{\h}{{\rm I}\kern-0.18em{\rm H}}
\newcommand{\K}{{\rm I}\kern-0.18em{\rm K}}
\newcommand{\p}{{\rm I}\kern-0.18em{\rm P}}
\newcommand{\E}{{\rm I}\kern-0.18em{\rm E}}
\newcommand{\Z}{{\rm Z}\kern-0.18em{\rm Z}}
\newcommand{\1}{{\rm 1}\kern-0.24em{\rm I}}
\newcommand{\N}{{\rm I}\kern-0.18em{\rm N}}

\newcommand{\x}{\mathcal{X}}

\newcommand{\pn}{\p_{\kern-0.25em n}}
\newcommand{\pnm}{\p_{\kern-0.25em n,m}}
\newcommand{\psubm}{\p_{\kern-0.25em m}}
\newcommand{\psubp}{\p_{\kern-0.25em p}}
\newcommand{\cfi}{\cF_{\kern-0.25em \infty}}

\newcommand{\argmax}{\mathop{\mathrm{argmax}}}

\newcommand{\eps}{\varepsilon}

\newlength{\minipagewidth}
\begin{document}

\begin{frontmatter}

\title{Online learning in repeated auctions.}
\runtitle{Learning in auctions}

 \author{ \fnms{Jonathan} \snm{Weed}\thanksref{t2, t3}\ead[label=jweed]{jweed@mit.edu},
 \fnms{Vianney} \snm{Perchet}\thanksref{t4}\ead[label=vianney]{vianney.perchet@normalesup.org} \and
 \fnms{Philippe} \snm{Rigollet}\thanksref{t2}\ead[label=rigollet]{rigollet@math.mit.edu}
}

\affiliation{Massachusetts Institute of Technology, Universit\'e Paris Diderot, and Massachusetts Institute of Technology}

\thankstext{t3}{Supported in part  by NSF Graduate Research Fellowship DGE-1122374.}
\thankstext{t2}{Supported in part by NSF grants DMS-1317308 and CAREER-DMS-1053987.}
\thankstext{t4}{Supported in part by ANR grant ANR-13-JS01-0004-01.}

\address{{Vianney Perchet}\\
{LPMA, UMR 7599}\\
{Universit\'e Paris Diderot}\\
{8, Place FM/13}\\
{75013, Paris, France}\\
\printead{vianney}
}

\address{{Philippe Rigollet}\\
{Department of Mathematics} \\
{Massachusetts Institute of Technology}\\
{77 Massachusetts Avenue,}\\
{Cambridge, MA 02139-4307, USA}\\
\printead{rigollet}
}

\address{{Jonathan Weed}\\
{Department of Mathematics} \\
{Massachusetts Institute of Technology}\\
{77 Massachusetts Avenue,}\\
{Cambridge, MA 02139-4307, USA}\\
\printead{jweed}
}

\runauthor{Weed et al.}

\begin{abstract}
Motivated by online advertising auctions, we consider repeated Vickrey auctions where goods of unknown value are sold sequentially and bidders only learn (potentially noisy) information about a good's value once it is purchased. We adopt an online learning approach with bandit feedback to  model this problem and derive bidding strategies for two models: stochastic and adversarial. In the stochastic model, the observed values of the goods are  random variables centered around the true value of the good. In this case, logarithmic regret is achievable when competing against well behaved adversaries. In the adversarial model, the goods need not be identical and we simply compare our performance against that of the best fixed bid in hindsight. We show that sublinear regret is also achievable in this case and prove matching minimax lower bounds. To our knowledge, this is the first complete set of strategies for bidders participating in auctions of this type.
\end{abstract}

\begin{keyword}[class=AMS]
\kwd[Primary ]{62L05}
\kwd[; secondary ]{62C20}
\end{keyword}
\begin{keyword}[class=KWD]
Second price auctions, Vickrey auctions, Repeated auctions, Bandit problems \end{keyword}

\end{frontmatter}

\section{Introduction}
\label{SEC:intro}

Online advertising has been a driving force behind most of the recent work on online learning, particularly in the realm of bandit problems. During the first quarter of 2015 alone, internet advertising generated $\$13.3$ billion in revenue, according to the Internet Advertising Bureau. A large fraction of advertising space is sold on platforms known as \emph{ad exchanges} such as Google's DoubleClick and AppNexus, which facilitate transactions between the owner of advertising space and advertisers. These transactions occur within a fraction of a second using auctions \cite{Mut09}, thus placing the actors squarely within the framework of game-theoretic auctions with a single item and multiple bidders.
In this context, we refer to the  advertising space as the \emph{good}, its owner as the \emph{seller} and the advertisers as \emph{bidders}, respectively. From the seller's perspective, this is a well understood problem in  mechanism design: the Vickrey (a.k.a.\ second price) auction is optimal in the sense that each bidder bidding their private value constitutes an equilibrium. Because of this property, the Vickrey auction is said to be \emph{truthful}.

The seller may also maximize her revenue while maintaining truthfulness of the auction by optimizing a \emph{reserve price} below which no transaction occur. For example, when the bidders' values are drawn independently from known distributions, the optimal reserve price may be computed in closed form~\cite{Mye81, RilSam81}. The independence assumption was questioned already by Myerson \cite{Mye81} and it was shown later \cite{CreMcL88} that when the assumption is violated, the seller can take advantage of the situation to extract more revenue at the cost of a more complicated auction mechanism. In particular, this mechanism allows bidders to be charged even if they do not win the auction, which is arguably undesirable. 

In short, the Vickrey auction is a reasonable compromise between simplicity and optimality, which likely explains its prevalence on ad exchanges. Nevertheless, it suffers from a major limitation: it relies on perfect knowledge of the bidders' value distributions, which are unlikely to be known to the seller in practice \cite{Wil87}. This limitation has driven a recent line of work on approximately optimal auctions \cite{RouTalYan12,HarRou09,FuHarHoy13} that are robust to misspecification of these distributions.  In recent years the ubiquitous collection of data has presented new opportunities, insofar as unknown quantities, such as the bidders' value distributions or relevant functionals, may potentially be learned from past observations. This new paradigm has been investigated in several recent papers: \cite{CesGenMan13, ChaHarNek14,FuHagHar14,OstSch11,ColRou14,AmiCumDwo15,KanNaz14,DhaRouYan15, BluManMor15, MohMed14, AmiRosSye14}.  
One of the take-home messages of this literature is that a few observations are sufficient to maximize the seller's revenue in the Vickrey auction. This not surprising since all that needs to be learned is the reserve price.

Strikingly, all the aforementioned work adopts the seller's perspective and focuses on designing mechanisms to maximize the seller's revenue. In this work, we take the perspective of a bidder engaged in repeated Vickrey auctions.  In the present paper, we identify and analyze several strategies that can be employed by a bidder in order to maximize his reward while simultaneously learning the value of a good sold repeatedly. This paradigm can be expressed as a learning problem with partial feedback, or \emph{bandit problem} \cite{BubCes12}. We are aware of only one other paper that takes the bidder's perspective~\cite{McA11} where using bandit strategies for bidding is suggested. 
Repeated auctions have been studied in the bandit framework, primarily in the context of \emph{truthful} bandits \cite{DevKak09,BabKleSli10,BabShaSli09}. However, this line of literature also takes the seller's point of view and aims at designing an auction mechanism rather than designing an optimal bidding strategy under the constraint of a simple mechanism such as the Vickrey auction.

More generally, the problem we describe falls into the category of \emph{partial monitoring games}, in which the learner receives only limited feedback about the loss associated with a given action. By analyzing the feedback structure of such games, it is possible to develop essentially optimal algorithms for many games in this class~\cite{BarFosPal14}. However, the performance guarantees of these algorithms degrade drastically as the number of actions increases. This renders these results unusable in our context, where the bidder's number of moves at each stage is essentially unbounded.

\section{Sequential Vickrey auctions}
\label{SEC:model}

We restrict our attention to bounded values and bids in the interval $[0,1]$.

Let us first recall the mechanism of a Vickrey auction for a single good. Each bidder $k \in [K+1]:=\{1,\ldots,K+1\}$ submits  a written bid $b[k]\in [0,1]$. The highest bidder $k^{\star} \in \argmax_k b[k]$ wins the auction and pays the second highest bid $m^{\star}=\max_{k \neq k^{\star}} b[k]$. In case of ties, the winner  is chosen uniformly at random among the highest bidders.

Each bidder $k \in [K+1]$ has a private but unknown individual value $v[k] \in [0,1]$, which represents the utility of the good. Note that this value is {independent} of the auction itself and is only measured by the bidder once the good is delivered to him. For example, in the case of advertising space, this value may be measured by the expected profit generated from this ad or the probability that it generates a click \cite{McA11}.  The reward of the winner is given by his \emph{net utility}  $v[{k^{\star}}]-m^{\star}$, while the reward of a loosing bidder is $0$.

Perhaps the most salient feature of the  Vickrey auction is that it is optimal for bidder $k$ to be \emph{truthful}, that is to bid $b[k] = v[k]$ (assuming that the bidder knows this value). Here optimality is understood in the equilibrium sense: any other bid $b[k] \neq v[k]$, even random, could never lead to a strict improvement in expected utility and might lead to a net loss for that bidder.
An implicit crucial assumption for the implementability of this bidding strategy is that each bidder must know his own value, a hypothesis that is not necessarily met in online repeated auctions. Nevertheless, a bidder may \textsl{learn} the value $v[k]$ from past observations. Like bandit problems, this problem exhibits an exploration-exploitation tradeoff: Higher bids  increase the number of observations and thus give the bidder a more accurate estimate of the value $v[k]$ (exploration) while bids closer to the best estimate of the value at time $t$ are more likely to be optimal in the sense described above (exploitation). We will see that auctions when viewed as bandit problems possess an idiosyncratic information feedback structure: information is collected only for higher bids, but these should be avoided due to the phenomenon known as the winner's curse~\cite{Wil69}.

We consider a set of $T \ge 2$ goods $t \in [T]:=\{1,\ldots,T\}$ that are sold sequentially in a Vickrey auction. Using a slight abuse of terminology, we will also call the auction at which good $t$ is sold auction $t$. We take the point of view of  bidder 1, hereafter referred to as \emph{the bidder}, and denote respectively by  $v_t, b_t, m_t \in [0,1]$  the unknown private value of  the bidder for the  $t^{\text{th}}$ good, his bid and the maximum bid of all other bidders for this good.   Without loss of generality\footnote{This can been achieved at an arbitrarily small cost by slightly perturbing original bids randomly.}, we assume that bids are never equal.  At time $t \ge 2$, the bidder is aware of the outcomes of past auctions\footnote{The bidder  knows $m_t$ for auctions that he won since it is the paid price, and we assume that the winning bid $m_t$ at times when he lost is made available publicly after each auction in order to incentivize higher future bids.}  $\{(b_s, m_s),\, s \in [t-1]\}$ as well as a (potentially noisy) measurement of the values of goods $[t-1]$ \emph{at times when the bidder won the auction}. Our goal is to construct bidding strategies that mitigate potential losses (overbidding) and opportunity cost (underbidding) for  the bidder.

We consider two  generating processes for   the sequence of values $\{v_t\}_t$: \emph{stochastic} and \emph{adversarial}.
The stochastic setup is the most benign one: consecutive values $\{v_t\}_t$  are independent and identically distributed (i.i.d.) random variables in the unit interval $[0,1]$.  On the other side of the stationarity spectrum is the adversarial setup, where the sequence $\{v_t\}_t$ may be any  sequence in $[0,1]$. This framework has become quite standard in the online learning literature  \cite{CesLug06, BubCes12} where a game-theoretic setup prevails and arbitrary dependencies between rounds occur.

\section{The stochastic setup}
\label{SEC:stoch}

Recall that consecutive values $\{v_t\}_t$  are independent and identically distributed (i.i.d.) random variables in the unit interval $[0,1]$. Let $v=\E[v_t]$ denote the common expected value of these random variables. It is easy to see that the expected net utility of the bidder at time t, $ \E(v_t-m_t)\mathds{1}\{ b_t > m_t \}$, is maximized at $b_t\equiv v$. Therefore, a constant bid equal to $v$ is optimal among all sequences of deterministic bids. This implies that the Vickrey auction is truthful in expectation. Since $v$ is unknown, the bidder may not be able to achieve the best net utility over $t$ rounds, so his performance is measured by his (cumulative) \textsl{pseudo-regret}\footnote{The benchmark in the (true) regret is the  \emph{random bid} that maximizes $b \mapsto  \sum_{t=1}^T(v_t-m_t)\mathds{1}\{ b > m_t \}$. This quantity is more difficulty to control and yields worse bounds, as detailed in Section~\ref{SEC:adversarial}.}  $\bar R_T$ defined by
\begin{equation}
\label{EQ:PseudoRegret}
\bar R_T= \max_{b \in [0,1]} \sum_{t=1}^T \E (v_t-m_t)\mathds{1}\{ b > m_t \}- \sum_{t=1}^T\E(v_t-m_t)\mathds{1}\{ b_t > m_t \},
\end{equation}
where the expectations are taken with respect to the randomness in $v_t$ and possibly in $m_t$ if the other bidders are playing randomly. Regret and pseudo-regret as measures of performance are studied primarily in the bandit literature but rarely in the context of auctions. Interestingly, using regret as a measure of performance allows us to take opportunity cost into account. Indeed, a net utility of zero can be obtained trivially at any round by bidding zero, but if the other bidders tend to bid below the value of the good, the regret will still scale linearly in $T$.

 \begin{wrapfigure}{L}{0.42\textwidth}
    \begin{minipage}{0.4\textwidth}
\begin{algorithm}[H]
\begin{algorithmic}
\STATE{\textbf{Input:} $b_1=1$, $\omega=1$, $\bar v=v_1$,}
\FOR{$t=2, \dots, T$}
\STATE{Bid $b_{t} =\min\Big( \overline{v} + \sqrt{\frac{3 \log t}{2\omega}}, 1\Big)$}
\STATE{Observe $m_t$}
\IF{$b_t>m_t$ (win auction)}
\STATE{Observe $v_t$}
\STATE{$\bar v\leftarrow (\omega\bar v + v_t)/(\omega+1)\,, \quad \omega\leftarrow\omega+1$}
\ENDIF{}
\ENDFOR
\end{algorithmic}
\caption{\aucb}
\label{UCBid}
\end{algorithm}
    \end{minipage}
  \end{wrapfigure}

We introduce a bidding strategy called \aucb because it is inspired by the {\sc UCB} algorithm \cite{LaiRob85, AueCesFis02} but tailored to the auction setup under investigation (See Algorithm~\ref{UCBid}). For the first auction, it prescribes to place the bid $b_1=1$ and thus win the auction. At auction $t+1$, $t \ge 1$, this strategy prescribes to place the bid $b_{t+1}$ defined by 
\[ b_{t+1} =\min\Big( \overline{v}_{\omega_t} + \sqrt{\frac{3 \log t}{2\omega_t}}, 1\Big)\, , 
\] 
where $\omega_t$ is the number of auctions won up to stage $t$ and $\overline{v}_{\omega_t}= \sum_{s=1}^{\omega_t} v_{\tau_s}/\omega_t$ with $\tau_s$ being the stage of the   $s^{\text{th}}$ won auction.  Interestingly, the \aucb strategy does not require the knowledge of past bids of other bidders $\{m_1, \dots, m_{t-1}\}$. This feature is particularly attractive in the setup of ad exchanges, where the process takes place so fast that it may be useful for the platform to not communicate the cost of an auction to bidders until the end of the day, for example. While the implementation of the \aucb strategy does not require the knowledge of $\{m_t\}_t$, its performance is affected by other bids that are larger but close to the optimal bid $v$. This is not surprising as such bids force the bidder to overpay in order to collect information about the unknown $v$. However, sub-linear regret of order $\sqrt{T}$ is achievable regardless of the sequence $\{m_t\}_t$. We prove two results that show that this strategy automatically \emph{adapts} to more favorable sequences $\{m_t\}_t$.

\subsection{Pseudo-regret bounds}

\begin{theorem}
\label{TH:sto1}
Consider the stochastic setup where the values $v_1, \ldots, v_T \in [0,1]$ are independent such that $\E[v_i]=v$. For any sequence $m_1, \ldots, m_T \in [0,1]$ such that $m_t$ is independent of $v_t$,  the \aucb strategy yields pseudo-regret bounded as follows:
$$
\bar R_T \le 3 +\frac{12\log T}{\Delta}\wedge 2\sqrt{6T\log T}\,,
$$
where $x\wedge y=\min(x,y)$ and  $\Delta \in [0, 1]$ is the largest number such that no bid $m_t$ is the interval $(v, v + \Delta)$.%
\end{theorem}
\begin{proof}
Since $v_t$ is independent of $(m_t,b_t)$ and $\E[v_t]=v$, we have
\begin{align*}
\bar R_T &= \max_{b \in [0,1]}\E\sum_{t=1}^T (v-m_t)\mathds{1}\{b>m_t\} - \E\sum_{t=1}^T (v-m_t)\mathds{1}\{b_t>m_t\}\\
&=\E\sum_{t=1}^T (v-m_t)\mathds{1}\{v>m_t\} - \E\sum_{t=1}^T (v-m_t)\mathds{1}\{b_t>m_t\}
\end{align*}
where in the second equality we used the fact that the supremum is attained at $b=v$ because  $(v-m_t)\mathds{1}\{ b > m_t \} \le (v-m_t)_{\tinyplus}=(v-m_t)\mathds{1}\{ v > m_t \}$, where $x_{\tinyplus}=\max(x, 0)$. 

Next, decomposing the regret on the the events $\{ b_t > m_t \}$ and $\{ b_t < m_t \}$, on which the bidder won and lost auction~$t$ respectively, we get
$$
(v-m_t)\mathds{1}\{v>m_t\} \le (v-m_t)\mathds{1}\{v>m_t>b_t\}+(v-m_t)_{\tinyplus}\mathds{1}\{b_t>m_t\}\,.
$$
This yields
\begin{align*} 
\bar R_T &\le \E\sum_{t=1}^T (v-m_t)_{\tinyplus}\mathds{1}\{b_t<v\} +\E \sum_{t=1}^T(m_t-v)_{\tinyplus}\mathds{1}\{v<m_t<b_t\}\nonumber\\
&\le \sum_{t=1}^T\p\{b_t<v\}+ \E\sum_{t=1}^T(m_t-v)\mathds{1}\{v<m_t<b_t\}\,.
\end{align*}
To control the first sum, using a union bound and Hoeffding's inequality, we get
$$
\p\{b_t<v\}\le \sum_{s=1}^t\p\big\{\bar v_s -v<- \sqrt{\frac{3\log t}{2s}} \big\}\le t^{-2}\,.
$$
so that
\begin{equation}
\label{EQ:reg_sto_1}
\bar R_T \le \frac{\pi^2}{6} + \E\sum_{t=1}^T(m_t-v)\mathds{1}\{v<m_t<b_t\}\,.
\end{equation}

Denote by $\omega_t$ the value of $\omega$ during the $t$th round. To control the second sum in~\eqref{EQ:reg_sto_1}, observe that, since $b_t>m_t$ implies that the bidder won auction $t$, we have $\omega_{t+1}=\omega_{t}+1$. Denote by $\cW=\{t\in [T]\,:\, b_t>m_t\}$ the set of auctions that the bidder has won. 
 If $m_t   \ge v+ \Delta$, we have
\begin{align*}
S&:=\E\sum_{t=1}^T(m_t-v)\mathds{1}\{v<m_t<b_t\}\\
&\le \E\sum_{t\in \cW}(m_t-v)\1\big\{\Delta < m_t -v< \bar v_{\omega_t}-v+ \sqrt{\frac{3 \log t}{2\omega_t}} \big\}\\
&\le  \sum_{t=1}^T\int_0^\infty\p\big\{\bar v_t+\sqrt{\frac{3 (\log T)}{2t}} -v>u+\Delta\big\}du\,.
\end{align*}
Using Hoeffding's inequality, we get
$$
\p\big\{\bar v_t-v>\sqrt{\frac{3 \log T}{2t}} +u \big\}\le T^{-3}e^{-u^2/2}\,.
$$
It yields, on the one hand, that for any $t \in [T]$,
\begin{align*}
\int_0^\infty\p\big\{\bar v_t+\sqrt{\frac{3 \log T}{2t}} -v>u+\Delta\big\}du&\le \sqrt{\frac{6 \log T}{t}}+ \int_0^\infty\p\big\{\bar v_t-v>\sqrt{\frac{3 \log T}{2t}} +u \big\}du\\
&\le  \sqrt{\frac{6 \log T}{t}}+ T^{-3}\sqrt{\frac{\pi}{2}}
\end{align*}
On the other hand, if  $t > t_\Delta:=6(\log T)/\Delta^2$, we have
\begin{align*}
\int_0^\infty\p\big\{\bar v_t+\sqrt{\frac{3 \log T}{2t}} -v>u+\Delta\big\}du&\le  \int_0^\infty\p\big\{\bar v_t-v>\sqrt{\frac{3 \log T}{2t}} +u \big\}du\le  T^{-3}\sqrt{\frac{\pi}{2}}\,.
\end{align*}
It yields
\begin{align*}
S\le \sum_{t=1}^{t_\Delta\wedge T} \sqrt{\frac{6 \log T}{t}}+\sqrt{\frac{\pi}{2}}\le \frac{12\log T}{\Delta}\wedge 2\sqrt{6T\log T}+\sqrt{\frac{\pi}{2}}\,.
\end{align*}
\end{proof}

Theorem~\ref{TH:sto1} shows an interesting phenomenon: While UCB type strategies are usually very sensitive to the assumption that the rewards are stochastic, this strategy is actually robust to \emph{any} sequence $\{m_t\}_t$ that may be generated by other bidders, including malicious ones, as long as $m_t$ is independent of the stochastic value $v_t$ for all $t$. Indeed, in this hybrid setup, where the $v_t$'s are random but the $m_t$'s may not be, the \aucb strategy exhibits a sublinear regret that can even be logarithmic in the favorable case where no bid $m_t$ is the interval $(v, v + \Delta)$ for some $\Delta>0$. It turns out that this condition can be softened and can be well captured by a simple margin condition under the assumption that the $m_t$'s are also stochastic. 

\subsection{Margin condition}

Assume in the rest of this section that $m_1, \ldots, m_T \simiid \mu$ for some unknown probability measure $\mu$. Borrowing terminology from binary classification~\cite{MamTsy99, Tsy06}, we define the \emph{margin condition} as follows.

\begin{definition}
A probability measure $\mu$ on $[0,1]$ satisfies the margin condition with parameter $\alpha>0$ around $v \in (0,1)$ if there exists a  constant $C_\mu>0$ such that 
$$
\mu\{(v,v+u]\} \le C_\mu u^{\alpha}\,\quad \forall\ u > 0\,.
$$
\end{definition}

The parameter $\alpha$ is an indication of the difficulty of the problem---the larger the $\alpha$, the easier the problem. Under the margin condition, we can interpolate between between the two bounds for the regret---$O(\log T)$ and  $O(\sqrt{T\log T})$---that arise in Theorem~\ref{TH:sto1}.

\begin{theorem}
\label{TH:sto2}
Fix $T \ge 2$ and consider the stochastic setup where the values $v_1, \ldots, v_T \in [0,1]$ are independent such that $\E[v_i]=v$. For any random sequence $m_1, \ldots, m_T \simiid \mu$, where $\mu$ on $[0,1]$ satisfies the margin condition with parameter $\alpha>0$ around $v \in (0,1)$,  the \aucb strategy yields pseudo-regret bounded as follows:
$$
\bar R_T \le \left\{
\begin{array}{lll}
\DS  c_1 T^{\frac{1-\alpha}{2}}\log^{\frac{1+\alpha}{2}} (T) &  \text{if } \alpha < 1\\
\DS  c_2 \log^2(T) &  \text{if } \alpha = 1\\
\DS  c_3 \log(T) &  \text{if } \alpha > 1
\end{array}\right.
$$
where $c_1, c_2$ and $c_3$ are positive constants that depend on $\alpha, v$ and $C_\mu$. 
\end{theorem}
\begin{proof}
We will prove the following bound:
$$
\bar R_T \le \left\{
\begin{array}{lll}
\DS & C_\mu\Big(\frac{12}{1-\alpha} T^{\frac{1-\alpha}{2}}\log^{\frac{1+\alpha}{2}} T +1\Big) &  \text{if } \alpha <1\\
\DS &   6C_\mu \Big(\log(T)+1\Big)^2 & \text{if } \alpha =1\\
\DS & 6 \log (T)\Big(1+ 2\frac{C_\mu}{\alpha\wedge 2-1}\Big)+\frac{4C_\mu}{\alpha\wedge 2-1}+1&  \text{if } \alpha >1
\end{array}\right.
$$
Recall from the proof of Theorem~\ref{TH:sto1} that
\begin{align*}
S&:=\E\sum_{t=1}^T(m_t-v)\mathds{1}\{v<m_t<b_t\}\\
&\le \E\sum_{t\in \cW}(m_t-v)\1\big\{0 < m_t -v< \bar v_{w_t}-v+ \sqrt{\frac{3 \log t}{2w_t}} \big\}\\
&\le  \E\sum_{t\in \cW}  \Big(\bar v_{w_t}-v+ \sqrt{\frac{3 \log T}{2w_t}}\Big)     \1\big\{0 < m_t -v< \bar v_{w_t}-v+ \sqrt{\frac{3 \log T}{2w_t}} \big\}\wedge 1\\
& \le \E\sum_{t=1}^T  \Big(\bar v_{t}-v+ \sqrt{\frac{3 \log T}{2t}}\Big)     \1\big\{0 < m_t -v< \bar v_{t}-v+ \sqrt{\frac{3 \log T}{2t}} \big\}\wedge 1\, ,
\end{align*}
where we used the fact that bids always belong to $[0,1]$. 
Using the margin condition, we get that for $\alpha \ge 0$
\begin{align*}
S&\le  C_\mu \E \sum_{t=1}^T \Big(\bar v_{t}-v+ \sqrt{\frac{3 \log T}{2t}}\Big)_{\tinyplus} ^{1+\alpha}\wedge 1 \,.
\end{align*}
Hoeffding's inequality yields that $\p\{ \bar v_t -v \ge \varepsilon\} \le e^{-2 t\varepsilon^2}$,  thus we get that
\begin{align*}
 \E  \Big(\bar v_{t}-v+ \sqrt{\frac{3 \log T}{2t}}\Big)_{\tinyplus} ^{1+\alpha} & \le (1+\alpha) \int_{- \sqrt{\frac{3 \log T}{2t}}}^\infty \Big(\varepsilon + \sqrt{\frac{3\log T}{2t}}\Big)^\alpha e^{-2t\varepsilon^2}d\varepsilon \\
 & \le \frac{1+\alpha}{t^{\frac{1+\alpha}{2}}} \int_{-\sqrt{\frac{3\log T}{2}}}^\infty \Big(s + \sqrt{\frac{3\log T}{2}}\Big)^\alpha e^{-2s^2}ds\\
 & \le \Big(\frac{6\log T}{t}\Big)^{\frac{1+\alpha}{2}}+\frac{1+\alpha}{2t^{\frac{1+\alpha}{2}}}\int_{\sqrt{6\log T}}^\infty u^\alpha e^{-u^2/2}du\, .
\end{align*}
As a consequence, if $\alpha \le 1$, we obtain
$$
S \le  C_\mu \sum_{t=1}^T \Big(\frac{6 \log T}{t}\Big)^{\frac{1+\alpha}{2}}+\frac{C_\mu}{T^2}  \,.
$$
and for $\alpha <1$, this yields that 
\begin{align*}
S&\le  C_\mu\Big(\frac{12}{1-\alpha} T^{\frac{1-\alpha}{2}}\log^{\frac{1+\alpha}{2}} T +1\Big)\, ,
\end{align*}
while, for $\alpha =1$, we get 
\begin{align*}
S&\le  6C_\mu \Big(\log(T)+1\Big)^2 \, .
\end{align*}

When $\alpha \leq 2$, it holds that $$\int_{\sqrt{6 \log T}}^\infty u^\alpha e^{-u^2/2}du \leq \int_{\sqrt{6 \log T}}^\infty u^2 e^{-u^2/2}du\leq2$$ hence
\begin{align*}
 S &\le 6 \log T+1+C_\mu \Big(\sum_{t=\lceil6\log T\rceil +1}^T \Big(\frac{6 \log T}{t}\Big)^{\frac{1+\alpha}{2}} +\frac{1+\alpha}{t^{\frac{1+\alpha}{2}}}\Big)\\
 &\le 6 \log (T)\Big(1+ 2\frac{C_\mu}{\alpha-1}\Big)+\frac{4C_\mu}{\alpha-1}+1\, . \end{align*}
For bigger values of $\alpha$, we shall use the fact that if the margin condition is satisfied for $\alpha\geq 2$, then it is also satisfied for the value $\alpha=2$. As a consequence, plugging the value $\alpha=2$ in the above equation, we obtain that
$$
S \leq 6\log(T)\Big(1+2C_\mu\Big)+ 4C_\mu+1\, .
$$
\end{proof}

As we can see from Theorem~\ref{TH:sto2}, the margin parameter $\alpha$ allows us to interpolate between $O(\log T)$ and $O(\sqrt{T})$ regret bounds. Since \aucb does not require the knowledge of $\alpha$, we say that it is \emph{adaptive to the margin parameter $\alpha$}.

In fact, the above result holds, with the exact same proof, under a weaker assumption. Denote by $\mu_t$ the  law of $m_t$ conditional on the past history $\{b_s,v_s,m_s\}_{s \leq t-1}$. Then the conclusions of Theorem~\ref{TH:exptreep} remain true if all $\mu_t$ satisfy the margin condition with respect to the same parameters $\alpha$ and $C_\mu$.

\subsection{Lower bound}
We now show that the family of rates---indexed by $\alpha$---is optimal up to logarithmic terms. As we shall see the upper bound is tight already in the case where the bid $\{m_t\}_t$ are i.i.d., independent of $\{v_t\}$. 

We first consider the case where $\alpha \in (0, 1)$. For any $\alpha$ in this interval, let $\mu_{\alpha}$ denote the distribution on $[0,1]$ with density $g_{\alpha}$ with respect to the Lebesgue measure, where $g_{\alpha}$ is defined by
$$
g_{\alpha}(x)=C_\alpha\Big[ \big(x-\frac12\big)^{\alpha-1}\1\big\{x\in (1/2,1/2+2\eps]\big\}+\big(x-\frac12-2\eps\big)^{\alpha-1}\1\big\{x\in (1/2+2\eps,1]\big\}\Big]\,,
$$
where $C_\alpha$ is an appropriate normalizing constant. (In what follows, $C_\alpha > 0$ is a constant that may change from line to line but depends on $\alpha$ only.)
See Figure~\ref{FIG:g} for a representation of this density. Observe that $\mu_{\alpha}$ satisfies the margin condition with parameter $\alpha>0$ around $v$. 

For $\alpha \geq 1$, define the distribution $\mu_\alpha$ to be the point mass at $1/2 + \eps$. This distribution also satisfies the margin condition with parameter $\alpha$.

\begin{figure}
\begin{tikzpicture}
  \draw[->] (-.5,0) -- (12,0) node[right] {$x$};
  \draw[->] (0,-.5) -- (0,4) node[above] {$g_\alpha(x)$};
  \draw[scale=1,domain=2.3:7.1,smooth,variable=\x,black] plot ({\x},{1/(\x-2)^.95});
\draw[scale=1,domain=7.3:12,smooth,variable=\x,black] plot ({\x},{1/(\x-7)^.95});
  \draw[-, scale=1, dashed] (2.1,-.3)node[below]{$\frac{1}{2}$}  -- (2.1,3.5) ;
  \draw[-, scale=1, dashed] (7.1,-.3)node[below]{$\frac{1}{2}+2\eps$}  -- (7.1,3.5) ;
    \draw[-, scale=1, red] (4.6,-.3)node[below]{$\frac{1}{2}+\eps$}  -- (4.6,3.5) ;
\end{tikzpicture}
\caption{Representation of the density $g_{.95}$ of bids $m_t$.}
\label{FIG:g}
\end{figure}
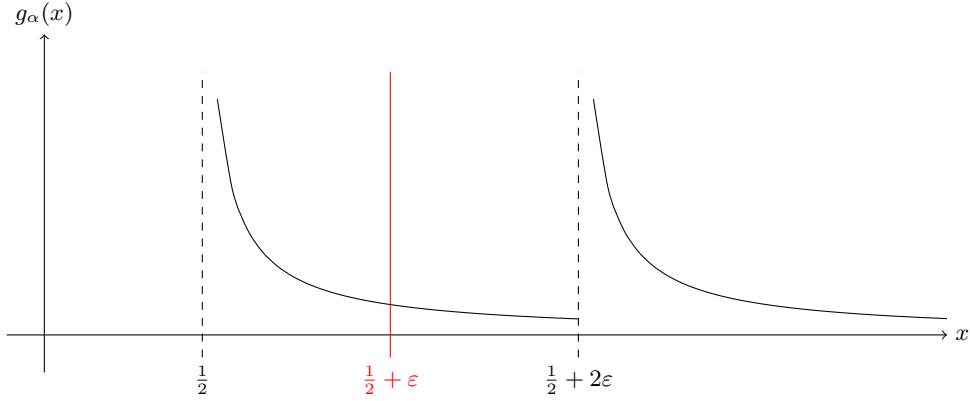

Let  $\nu$ denote the joint distribution of $(v_t, m_t)$ and denote by $\bar R_T(\nu)$ the pseudo-regret associated to a strategy when the expectation in~\eqref{EQ:PseudoRegret} is taken with respect to $\nu$. 

\begin{theorem}
\label{TH:sto_lb}
Fix $\alpha > 0$. Let $\nu=\bern(1/2)\otimes \mu_{\alpha}$ and $\nu'=\bern(1/2+2\eps)\otimes \mu_{\alpha}$, where $\eps=\frac 1 2 T^{-1/2}$. Then, for any strategy, it holds
$$
\bar R_T(\nu)\vee \bar R_T(\nu')\ge 
\left\{\begin{array}{ll}
C_\alpha T^{\frac{1-\alpha}{2}} & \text{if $\alpha < 1$} \\
C_\alpha \log  T & \text{if $\alpha \geq 1$}
\end{array}
\right.
$$
\end{theorem}
\begin{proof}
We first consider the case where $\alpha < 1$. Recall from~\eqref{EQ:PseudoRegret} that the pseudo-regret is given by $\bar R_T=\E\sum_{t=1}^Tr_t$
where $r_t$ denotes the \emph{instantaneous regret}, defined by
$$
r_t(\nu)=\E_\nu(v-m_t)\mathds{1}\{v>m_t\} - \E_\nu (v-m_t)\mathds{1}\{b_t>m_t\}\,.
$$
Note first that under $\nu$ or $\nu'$ we can restrict our attention to strategies that bid $b_t\ge 1/2$.  Observe first that since $v=1/2$ under $\nu$, the definition of the pseudo-regret~\eqref{EQ:PseudoRegret} simplifies to
\begin{equation}
\label{EQ:pr:lbsto:1}
\bar R_T(\nu)=\sum_{t=1}^T\E_\nu(m_t-v)\1\{b_t >m_t\}=\E_\nu\sum_{t=1}^T\int_{1/2}^{b_t}(x-1/2)g_\alpha(x)dx%
\end{equation}
Moreover,
$$
\int_{1/2}^{b_t}(x-1/2)g_\alpha(x)dx \ge C_\alpha \big[\bar b_t^{\alpha+1}\1\{\bar b_t \le 2\eps\} + \big((2\eps)^{\alpha+1}+(\bar b_t -2\eps)^{\alpha+1}\big)\1\{ \bar b_t > 2\eps\}  \big]
$$
where  $\bar b_t=b_t-1/2\ge 0$.  Therefore
$$
\bar R_T(\nu)\ge C_\alpha \E_\nu S_{\alpha+1}\,,
$$
where 
$$
S_\alpha= \sum_{t=1}^T\bar b_t^{\alpha}\1\{\bar b_t \le 2\eps\} + \big((2\eps)^{\alpha}+(\bar b_t -2\eps)^{\alpha}\big)\1\{ \bar b_t > 2\eps\}\,.
$$
We will use the fact that $\E_\nu S_{\alpha+1} \ge (2\eps)^{\alpha+1}\mathsf{S}(\eps)$ and $\E_\nu S_\alpha \le (2\eps)^{\alpha}T+ \mathsf{S}(\eps)$, where
$$
\mathsf{S}(\eps)=\sum_{t=1}^T\p_{\nu}\{ \bar b_t > 2\eps\}\,.
$$

Next, for any strategy, define the associated test $\psi_t\in \{\nu , \nu'\}$ by $\psi=\nu$ if $\bar b_t\le \eps$ and $\psi=\nu'$ if $\bar b_t > \eps$. 
One the one hand, under $\nu$, the instantaneous regret $r_t$ satisfies
$$
r_t(\nu)\ge \E_\nu(m_t-1/2)\1\{1/2+\eps>m_t\}\1\{b_t>1/2+\eps\} \ge C_\alpha\eps^{\alpha+1}\p_\nu(\psi_t=\nu')\,.
$$
On the other hand, under $\nu'$,  the instantaneous regret $r_t$ satisfies
\begin{align*}
r_t(\nu')&\ge \E_{\nu'}\big[\1\{b_t<1/2+\eps\}(1/2+2\eps-m_t)\big(\1\{1/2+2\eps>m_t\}-\1\{b_t>m_t\}\big)\big]\\
&\ge \E_{\nu'}\big[\1\{b_t<1/2+\eps\}(1/2-2\eps-m_t)\big(\1\{1/2+2\eps>m_t\}-\1\{1/2+\eps>m_t\}\big)\big]\\
&= \E_{\nu'}\big[\1\{b_t<1/2+\eps\}(1/2-2\eps-m_t)\1\{1/2+\eps\le m_t <1/2+2\eps\big]\\
&= C_\alpha\p_{\nu'}(\psi_t=\nu)\int_\eps^{2\eps}x^\alpha(2\eps-x)dx\ge C_\alpha\p_{\nu'}(\psi_t=\nu) \eps^{\alpha+1}
\end{align*}
The last two displays yield
\begin{equation}
\label{EQ:pr:lbsto:2}
r_t(\nu) + r_t(\nu')\ge C_\alpha\eps^{\alpha+1}\big[\p_\nu(\psi_t=\nu')+\p_{\nu'}(\psi_t=\nu)\big]\,.
\end{equation}
It follows from Sanov's inequality (see, e.g., \cite{BubPerRig13}, Lemma~4) that
$$
\p_\nu(\psi_t=\nu')+\p_{\nu'}(\psi_t=\nu) \ge \frac{1}{2}\exp\big[-\KL(\nu^{\otimes t}, \nu'^{\otimes t})\big]\,.
$$
Moreover, since (i) $m_t$ has the same distribution under both $\nu$ and $\nu'$ and (ii), $v_t$ is observed only when $b_t\ge m_t$, we get
\begin{align*}
\KL(\nu^{\otimes t}, \nu'^{\otimes t})&=\E_{\nu}\sum_{s=1}^t\1(b_t \ge m_t)\KL\big(\bern(1/2), \bern(1/2+2\eps)\big)\\
&\le 4\eps^2\sum_{s=1}^t\p_{\nu}(m_t \le b_t)\\
&\le C_\alpha\eps^2 \E_\nu S_\alpha\\
&\le C_\alpha (2\eps)^{2+\alpha}T+\eps^2 \mathsf{S}(\eps)
\end{align*}
where we used the fact that $\eps \le (2\sqrt{2})^{-1}$ in the first inequality. Together with~\eqref{EQ:pr:lbsto:1} and~\eqref{EQ:pr:lbsto:2}, the above two displays yield
\begin{align*}
\bar R_T(\nu) + \bar R_T(\nu')&\ge C_\alpha \big(T\eps^{\alpha+1}\exp\big[-C_\alpha \big((2\eps)^{2+\alpha}T+\eps^2 \mathsf{S}(\eps)\big)\big]+(2\eps)^{\alpha+1}\mathsf{S}(\eps)\big)\\
&\ge C_\alpha \big(T\eps^{\alpha+1}\exp\big[-\eps^2 \mathsf{S}(\eps)\big]+(2\eps)^{\alpha+1}\mathsf{S}(\eps)\big)
\end{align*}
for $\eps$ such that $(2\eps)^{2+\alpha}T\le 1$. 
We obtain
$$
\bar R_T(\nu) + \bar R_T(\nu')\ge C_\alpha \inf_{S \in [0,T]}\big\{T\eps^{\alpha+1}\exp(-\eps^2 S) + \eps^{\alpha+1}S\big\} 
$$
The infimum is achieved when 
$$
S=\frac{\log(T\eps^2)}{\eps^2} \vee 0\,,
$$
so
\begin{equation*}
\bar R_T(\nu) + \bar R_T(\nu')\ge C_\alpha \left[\left( \eps^{\alpha -1} + \eps^{\alpha -1} \log(T\eps^2)\right) \vee T \eps^{\alpha + 1}\right]
\end{equation*}
for all $\eps \leq \frac 1 2 T^{-1/(2+\alpha)}$. Since $\alpha < 1$, we can choose $\eps = \frac{1}{2}T^{-1/2}$, in which case 
\begin{equation*}
\bar R_T(\nu) + \bar R_T(\nu')\ge C_\alpha T^{\frac{1-\alpha}{2}}\,,
\end{equation*}
as desired.

When $\alpha \geq 1$, we obtain the following analogue to~\eqref{EQ:pr:lbsto:1}:
\begin{equation*}
\bar R_t(\nu) = \sum_{t=1}^T\E_\nu\eps\1\{b_t > 1/2+\eps\} = \eps \mathsf{S'}(\eps)\,
\end{equation*}
where 
$$
\mathsf{S'}(\eps) = \sum_{t=1}^T \p_\nu \{b_t > 1/2 + \eps\}\,.
$$
The rest of the proof is the same apart from some small changes. Since $v_t$ is only observed when $b_t \geq 1/2 + \eps$, we obtain the bound
\begin{equation*}
\KL(\nu^{\otimes t}, \nu'^{\otimes t}) \leq C_\alpha \eps^2 \mathsf{S'}(\eps)\,.
\end{equation*}
This yields
\begin{equation*}
\bar R_T(\nu) + \bar R_T(\nu')\ge C_\alpha \inf_{S \in [0,T]}\big\{T\eps\exp(-\eps^2 S) + \eps S\big\}\,.
\end{equation*}
The infimum is attained at the same value of $S$, which implies 
\begin{equation*}
\bar R_T(\nu) + \bar R_T(\nu')\ge C_\alpha \left[\left( 1 + \log(T\eps^2)\right) \vee T \eps^2\right]
\end{equation*}
for all $\eps < 1/4$. Choosing $\eps = O(1)$ yields the claim.
\end{proof}

\section{The adversarial setup}
\label{SEC:adversarial}

In this section, unlike the stochastic case, we make no assumptions on the sequences $\{v_t\}_t$ and $\{m_t\}_t$, even allowing the seller and other bidders to coordinate their plays according to a non-stationary process. As in the stochastic case, we compare the performance of a sequence $\{b_t\}_t$ of bids generated by a data-driven strategy to the \emph{best fixed bid} in hindsight. 
As a consequence the (cumulative) \textsl{regret} $R_T$ of the bidder for not knowing his own sequence of values is defined as
\begin{equation}
\label{EQ:Regret}
R_T= \max_{b \in [0,1]} \sum_{t=1}^T  (v_t-m_t)\mathds{1}\{ b > m_t \}- \sum_{t=1}^T(v_t-m_t)\mathds{1}\{ b_t > m_t \}\,.
\end{equation}
As in the stochastic case, we will also consider the pseudo-regret $\bar R_T$, defined in (\ref{EQ:PseudoRegret}), which is easier to handle and will serve as an an illustration of the techniques used in our proofs.

Clearly, $\bar R_T \le \E[R_T]$ and it is well known that $\bar R_T = \E[R_T]$ when the adversary is oblivious  \cite{CesLug06, BubCes12}, that is, when it generates its sequence of moves independently of the past actions of the bidder. In the sequel, we study both oblivious and non-oblivious (a.k.a.\ adaptive) adversaries.

We henceforth consider a shifted version of the auction described above where the reward associated to bid $b$ at time $t$ is given by
\begin{equation*}
g(b, t) = (v_t - m_t) \indic{b > m_t} + m_t\,.
\end{equation*}
Shifting the reward of the game in this way does not affect the regret, but it has the convenient effect that the bidder's net utility at each round is positive. 

For notational convenience, assume hereafter that $m_t \in (0,1]$ and that $v_t \in [0,1]$. Precluding $m_t=0$ has no effect on the problem if we replace $m_t=0$ by an arbitrarily small value.
\newpage

\subsection{Oblivious adversaries}
 \begin{wrapfigure}{r}{0.45\textwidth}
 \vspace{-0.6cm}
    \begin{minipage}{0.45\textwidth}
\begin{algorithm}[H]
\begin{algorithmic}
\STATE{\textbf{Input:} $\eta \in (0, 1/2)$, $\cL=(0,1]$, $w_{(0,1]}=1$, $p_{(0,1]}=1$\,.}
\FOR{$t = 1, \dots, T$}
\STATE{Select $\ell \in \cL$ with probability $p_{\ell}$ and $b \sim \unif(\ell)$}
\STATE{Bid
$b_t =\left\{ \begin{array}{ll} 1 & \text{with probability $\eta$} \\
0 & \text{with probability $\eta$} \\
b & \text{with probability $1-2\eta$} \end{array} \right.$}
\STATE{Observe $m_t \in \bar\ell=(x,y]$} 
\STATE{$\bar \ell_l=(x,m_t]$, $\bar \ell_r= (m_t, y]$}
\STATE{$w_{\bar \ell_l}\leftarrow w_{\bar \ell}$, $w_{\bar \ell_r}\leftarrow w_{\bar \ell}$}
\STATE{$\cL\leftarrow (\cL\setminus \bar \ell)\cup  \bar \ell_l \cup \bar \ell_r$}
\FOR{$\ell \in \cL$}
\IF{$b_t>m_t$}
\STATE{Observe $v_t$}
\STATE{$\hat g(\ell) \leftarrow \frac{v_t\1\{m_t\preceq \ell\}}{\p_{ \cB_t} (b_t > m_t)}$}
\ELSE
\STATE{$\hat g(\ell) \leftarrow  \frac{m_t\indic{m_t \succeq \ell}}{1-\p_{ \cB_t} (b_t > m_t)}$}
\ENDIF
\STATE{$w_\ell \leftarrow w_{\ell} \exp(\eta \hat  g(\ell))$\,, \quad $p_{\ell} \leftarrow \frac{|\ell| w_{\ell, t}}{\sum_{\ell \in \cL} |\ell| w_{\ell}}$}
\ENDFOR
\ENDFOR
\end{algorithmic}
\caption{\exptree}
\label{alg:exptree}
\end{algorithm}
\end{minipage}
  \end{wrapfigure}
One popular strategy for adversarial partial-information problems of this kind is the celebrated \expthree{} algorithm~\cite{AueCesFre01}. However, \expthree{} and similar approaches are tailored to problems with a \emph{fixed} number of actions. In the auction setup, by contrast, the number of actions is \emph{a priori} unbounded, and even the number of actions up to equivalence grows with $T$.  Standard tools are therefore incapable of achieving sublinear regret in this regime. In Algorithm~\ref{alg:exptree}, we present a novel strategy for bandit games of this type that allows the number of actions to grow over time.

The algorithm maintains a sequence of nested partitions  $\cL_t, t \geq 1$ of $(0, 1]$ into $t$ intervals of the form $(x,y]$ for $0\le x<y\le 1$. We set $\cL_1=(0,1]$ and the refinement of the partition $\cL_t$ is done as follows. Let $\bar\ell=(x,y] \in \cL_t$ be the unique interval in $\cL_t$ such that $m_t \in \bar\ell$. Then $\bar\ell$ is \emph{split} into two subintervals $\bar\ell_l=(x, m_t]$ and $\bar\ell_r=(m_t,y]$: $\cL_{t+1}=(\cL_t\setminus\bar\ell)\cup\bar\ell_l\bar\cup \ell_r$. This procedure is illustrated in Figure \ref{Fig:Split}.

\begin{figure}[h!]
\begin{center}
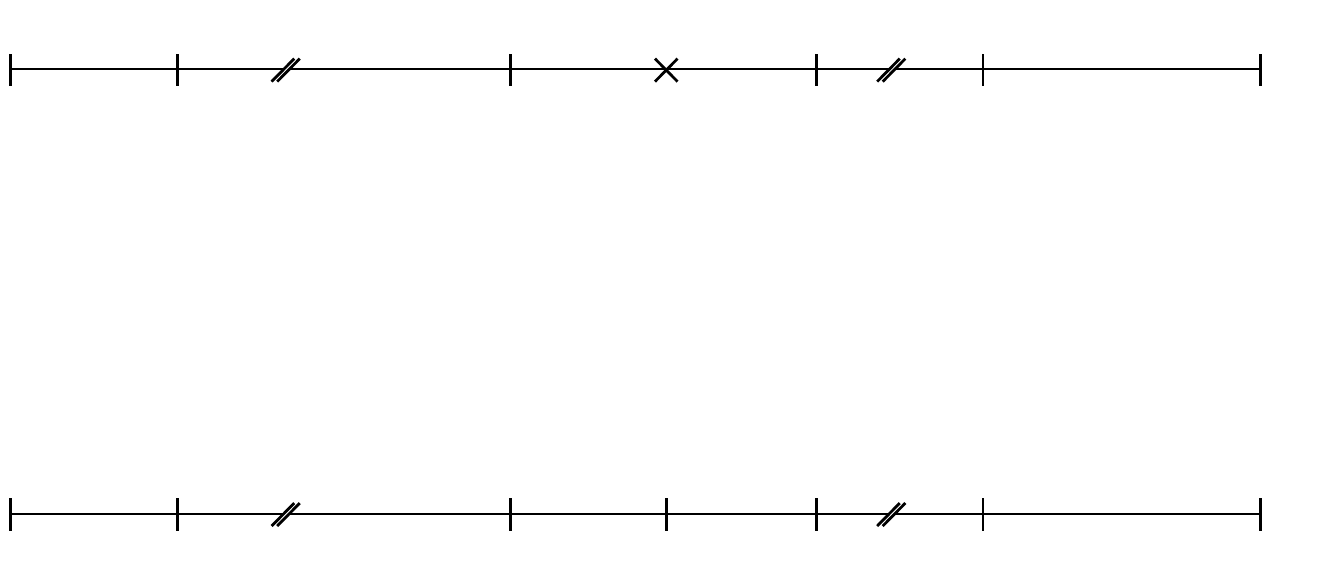
\end{center}
\caption{Illustration of the splitting procedure for constructing $\cL_{t+1}$ from $\cL_t$}
\label{Fig:Split}
\end{figure}

Each element $\ell \in \cL_t$ is assigned a probability $p_{\ell,t}$ defined in~\eqref{EQ:expweights} below and such that $p_{\ell,t}>0$ and $\sum_{\ell \in \cL_t} p_{\ell,t}=1$. At round $t$, the \exptree{} strategy prescribes to bid randomly as follows. With constant probability, bid $0$ or $1$. Otherwise, first draw $\ell \in \cL_t$ with probability $p_{\ell,t}$ and then draw a bid $b_t\sim \unif(\ell)$ uniformly over the interval $\ell$. We denote the resulting distribution of $b_t$ by $\cB_t$ and by $\p_{\cB_t}$ the associated probability. Note that $\cB_t$ is a mixture of uniform distributions that can be computed explicitly given $p_{\ell,t}, \ell \in \cL_t$:
\begin{equation}
\label{EQ:distBt}
\p_{\cB_t}(A)=(1-2\eta)\sum_{\ell \in \cL_t} p_{\ell,t}|A\cap \ell|\,, \qquad \forall A \in (0,1)\ \text{mesurable,}
\end{equation}
where here and in what follows, $|A|$ denotes the Lebesgue measure of $A \subset [0,1]$.
It remains only to specify the distribution $p_{\ell,t}, \ell \in \cL_t$. Intuitively, we hope to construct this distribution based on the intervals' past performance, but since the player only observes the value $v_t$ when $b_t > m_t$, we cannot evaluate the gain $g(b, t)$ of an arbitrary bid $b$ at round $t$. Instead, we compute an unbiased estimate $\hat g(b, t)$  of $g(b,t)$ given by
$$
\hat g(b, t)  =\frac{v_t \1\{b_t > m_t\}}{\p_{\cB_t} (b_t > m_t)}\1\{b > m_t\} + \frac{m_t\1\{b_t \le m_t\}}{1-\p_{\cB_t} (b_t > m_t)}\1\{b \leq m_t\}\,.
$$
It is not hard to check that $\E_{b_t\sim\cB_t}[\hat g(b, t)]=g(b,t)$. Moreover, this estimate is constant on each interval $\ell \in \cL_{t+1}$ and depends only on whether   $m_t\preceq \ell$ (i.e., $m_t \le x$ for all $x \in \ell$) or $m_t\succeq \ell$. As a result, overloading the notation,  we define the following estimate for the gain of a bid in the interval $\ell$:
\begin{equation}
\label{EQ:estgain1}
\hat g(\ell, t)=\frac{v_t \1\{b_t > m_t\}}{\p_{\cB_t} ( b_t> m_t)}\1\{m_t\preceq \ell\}  + \frac{m_t\1\{b_t \le m_t\}}{1-\p_{\cB_t} (b_t > m_t)}\1\{m_t\succeq \ell \}\,.
\end{equation}

With this estimate, we can compute $p_{\ell,t+1}, \ell \in \cL_{t+1}$ using exponential weights:
\begin{equation}
\label{EQ:expweights}
p_{\ell,{t+1}}=\frac{|\ell|w_{\ell,t+1}}{\sum_{\kappa \in \cL_{t+1}}|\kappa|w_{\kappa,t+1}}\,, \quad
w_{\ell,t}=\exp\Big(\eta \sum_{s=1}^{t-1}\hat g(\ell, s)\Big)\,, \ell \in \cL_{t+1}
\end{equation}
for some tuning parameter $\eta>0$ to be chosen carefully. The reweighing by the length $|\ell|$ of the interval $\ell$ in~\eqref{EQ:expweights} is the main novelty of our algorithm. 

\begin{theorem}
\label{TH:exptree}
Let $v_1, \dots, v_T \in [0,1]$ and $m_1, \dots, m_T \in [0,1]$ be arbitrary sequences. Let $\ell^{\circ} \in \cL_T$ denote the widest interval in the finest partition $\cL_T$ such that  $\argmax_{b \in [0, 1]} \sum_{t=1}^T \E(v_t - m_t)\indic{b > m_t} \cap \ell^\circ \neq \emptyset$ and let $\Delta^\circ=|\ell^\circ|$ denote its width. The strategy \exptree{} run with parameter
$\DS \eta = (1/2) \sqrt{\log(1/\Delta^{\circ})/T}\wedge(1/2)$
achieves the pseudo-regret bound
\begin{equation}
\bar R_T \leq 4\sqrt{T \log(1/\Delta^{\circ})}.
\end{equation}
\end{theorem}
\begin{proof}
Any choice $\eta \le 1/2$ guarantees that the probability distribution $\cB_t$ constructed above is valid. Moreover, when $\sqrt{\log(1/\Delta^{\circ})/T} > 1$, the claimed bound is vacuous, so we can assume that $\eta = (1/2)\sqrt{\log(1/\Delta^{\circ})/T}.$

Define $W_t=\sum_{\kappa \in \cL_t}|\kappa| w_{\kappa,t}$.
By extending the definition~\eqref{EQ:expweights} of $p_{\ell,t}$ to all $\ell \in \cL_{t+1}$, we can write
\begin{align*}
\log \frac{W_{t+1}}{W_{t}} & = \log \sum_{\ell \in \cL_{t+1}} \frac{|\ell|w_{\ell, t} \exp(\eta \hat g(\ell, t))}{W_t}  = \log \sum_{\ell \in \cL_{t+1}} p_{\ell, t} \exp(\eta\hat g(\ell, t)).
\end{align*}
By construction, $\eta\hat g(\ell, t) \leq 1$. Since $e^x \leq 1 + x + x^2$ for $x \leq 1$, this implies
\begin{align}
\label{EQ:pr:exptree1}
\log \frac{W_{t+1}}{W_{t}}  & \leq \log \sum_{\ell \in \cL_{t+1}} p_{\ell, t} (1 + \eta \hat g(\ell, t) + \eta^2 \hat g(\ell, t)^2)  \\
& = \log \big( 1+ \eta \sum_{\ell \in \cL_{t+1}} p_{\ell, t} \hat g(\ell, t) + \eta^2 \sum_{\ell \in \cL_{t+1}} p_{\ell, t} \hat g(\ell, t)^2\big) \nonumber \\
& \leq \eta \sum_{\ell \in \cL_{t+1}} p_{\ell, t} \hat g(\ell, t) + \eta^2 \sum_{\ell \in \cL_{t+1}} p_{\ell, t} \hat g(\ell, t)^2\,. \nonumber
\end{align}
It follows from~\eqref{EQ:distBt} that $\p_{\cB_t}(b_t \in \ell) = (1-2\eta) p_{\ell, t}$ for all $\ell \in \cL_{t}$. Moreover, for $\ell_l, \ell_r \in  \cL_{t+1}\setminus \cL_{t}$, there exists $\ell \in \cL_t$ such that $\ell=\ell_l\cup \ell_r$ and
$$
\p_{\cB_t}(b_t \in \ell_r)=\frac{|\ell_r|}{|\ell|} \p_{\cB_t}(b_t \in \ell)= (1-2\eta) {p_{\ell_r, t}}\,.
$$
 Of course, the same holds for $\ell_l$ so that
 $\p_{\cB_t}(b_t \in \ell) = (1-2\eta) p_{\ell, t}$ for all $\ell \in \cL_{t+1}$.
Moreover, it follows from~\eqref{EQ:estgain1} that
\begin{align*}
\sum_{\ell \in \cL_{t+1}} \p_{\cB_t}(b_t \in \ell) \hat g(\ell, t) &=\sum_{\stackrel{\ell \in \cL_{t+1}}{ m_t \preceq \ell} }\frac{\p_{\cB_t}(b_t \in \ell)}{\p_{\cB_t} ( b_t> m_t)} v_t\1(b_t>m_t)+\sum_{\stackrel{\ell \in \cL_{t+1}}{m_t \succeq \ell}} \frac{\p_{\cB_t}(b_t \in \ell)}{1-\p_{\cB_t} ( b_t> m_t)} m_t\1(b_t\le m_t)\\ & \le g(b_t,t)\,.
\end{align*}
Since $g(b_t, t) \leq 1$, we also have
\begin{align*}
\sum_{\ell \in \cL_{t+1}} \p_{\cB_t}(b_t \in \ell) \hat g(\ell, t)^2&=\sum_{\stackrel{\ell \in \cL_{t+1}}{ m_t \preceq \ell} }\frac{\p_{\cB_t}(b_t \in \ell)}{\p_{\cB_t} ( b_t> m_t)} v_t\1\{b_t>m_t\}\hat g(\ell, t)\\&\phantom{=}+\sum_{\stackrel{\ell \in \cL_{t+1}}{m_t \succeq \ell}} \frac{\p_{\cB_t}(b_t \in \ell)}{1-\p_{\cB_t} ( b_t> m_t)} m_t\1\{b_t\le m_t\}\hat g(\ell, t)\\
&\leq \frac{v_t^2\1(b_t>m_t)}{\p_{\cB_t} ( b_t> m_t)}+ \frac{m_t^2\1(b_t\le m_t)}{\p_{ \cB_t} ( b_t> m_t)}= g(b_t, t) \hat g(b_t,t) \leq \hat g(b_t, t)\,.
\end{align*}
These two inequalities yield respectively
\begin{align*}
\sum_{\ell \in \cL_{t+1}} p_{\ell, t} \hat g(\ell, t) = \frac{1}{1-2\eta} \sum_{\ell \in \cL_{t+1}} \p_{\cB_t}(b_t \in \ell) \hat g(\ell, t) \le \frac{1}{1-2\eta}g(b_t, t)\,,
\end{align*}
and
$$
\sum_{\ell \in \cL_{t+1}} p_{\ell, t} \hat g(\ell, t)^2 = \frac{1}{1-2\eta} \sum_{\ell \in \cL_{t+1}} \p_{\cB_t}(b_t \in \ell) \hat g(\ell, t)^2 \leq \frac{1}{1-2\eta} \hat g(b_t, t)\,,
$$

Combining the above two displays with~\eqref{EQ:pr:exptree1} yields
\begin{equation*}
\log \frac{W_{t+1}}{W_{t}} \leq \frac{\eta}{1-2\eta} g(b_t, t) + \frac{\eta^2}{1-2\eta}  \hat g(b_t, t)\,.
\end{equation*}
It follows from (\ref{EQ:estgain1}) that $\E[\hat g(b_t, t)] = v_t + m_t \leq 2$, hence
\begin{equation*}
\E\big[\log \frac{W_{t+1}}{W_{t}}\big] \leq \frac{\eta}{1-2\eta} \E g(b_t, t) + \frac{2 \eta^2}{1-2\eta}\,.
\end{equation*}

Let $G(b) = \sum_{t=1}^T \hat g(b, t)$ and $\bar G = \sum_{t=1}^T g(b_t, t)$. Summing on $T$, we obtain
\begin{align*}
\E\big[\log \frac{W_T}{W_0}\big] &\leq \frac{\eta}{1-2\eta} \E \bar G + \frac{2 \eta^2T}{1-2\eta}\,.
\end{align*}

Let $b^{\circ} \in \argmax_{b \in [0, 1]} \sum_{t=1}^T \E(v_t - m_t)\indic{b > m_t}$, and suppose $b^{\circ} \in \ell^{\circ} \in \cL_T$. We can bound $W_T$ by writing
$$
\E[\log W_T] = \E\big[\log \sum_{\ell \in \cL_T} |\ell|\exp\big(\eta \sum_{t=1}^T\hat g(\ell, t)\big)\big] \geq \E\big[\log |\ell^\circ|\exp\big(\eta \sum_{t=1}^T\hat g(b^{\circ}, t)\big)\big] = \log \Delta^{\circ} + \eta G(b^{\circ}).
$$

Rearranging and noting that $W_0 = 1$, we obtain
\begin{equation*}
(1- 2\eta) \max_{b \in [0, 1]} \E G(b) -  \E \bar G \leq 2 \eta T + \frac{(1-2\eta)\log(1/\Delta^{\circ})}{\eta} \leq 2\eta T + \frac{\log(1/\Delta^{\circ})}{\eta}.
\end{equation*}

Finally, since $G(b) \leq T$, we obtain
\begin{equation*}
\bar R_T \leq 4 \eta T + \frac{\log(1/\Delta^{\circ})}{\eta}.
\end{equation*}

Plugging in the given value of $\eta$ yields the claim.
\end{proof}

 \begin{wrapfigure}{R}{0.6\textwidth}
    \begin{minipage}{0.6\textwidth}
\begin{algorithm}[H]
\begin{algorithmic}
\STATE{\textbf{Input:} $\eta \in (0, 1/8)$, $\gamma \in (0, 1/4)$, $\beta \in (0, 1)$, $\cL=(0,1]$, $w_{(0,1]}=1$, $p_{(0,1]}=1$\,.}
\FOR{$t = 1, \dots, T$}
\STATE{Select $\ell \in \cL$ with probability $p_{\ell}$ and $b \sim \unif(\ell)$}
\STATE{Bid
$b_t =\left\{ \begin{array}{ll} 1 & \text{with probability $\gamma$} \\
0 & \text{with probability $\gamma$} \\
b & \text{with probability $1-2\gamma$} \end{array} \right.$}
\STATE{Observe $m_t \in \bar\ell=(x,y]$ and define $\bar \ell_l=(x,m_t]$, $\bar \ell_r= (m_t, y]$}
\STATE{$w_{\bar \ell_l}\leftarrow w_{\bar \ell}$, $w_{\bar \ell_r}\leftarrow w_{\bar \ell}$}
\STATE{$\cL\leftarrow (\cL\setminus \bar \ell)\cup  \bar \ell_l \cup \bar \ell_r$}
\FOR{$\ell \in \cL$}
\IF{$b_t>m_t$}
\STATE{Observe $v_t$}
\STATE{$\tilde g(\ell) \leftarrow \frac{v_t+\beta}{\p_{ \cB_t} (b_t > m_t)}\1\{m_t\preceq \ell\} + \frac{\beta}{1 - \p_{ \cB_t} (b_t > m_t)} \1\{m_t \succeq \ell\}$}
\ELSE
\STATE{$\tilde g(\ell) \leftarrow  \frac{\beta}{\p_{ \cB_t}(b_t > m_t)}\1\{m_t\preceq \ell\} + \frac{m_t+\beta}{1-\p_{ \cB_t} (b_t > m_t)}\indic{m_t \succeq \ell}$}
\ENDIF
\STATE{$w_\ell \leftarrow w_{\ell} \exp(\eta \tilde g(\ell))$\,, \quad $p_{\ell} \leftarrow \frac{|\ell| w_{\ell, t}}{\sum_{\ell \in \cL} |\ell| w_{\ell}}$}
\ENDFOR
\ENDFOR
\end{algorithmic}
\caption{\exptreep}
\label{alg:exptreep}
\end{algorithm}
\end{minipage}
  \end{wrapfigure}
Note that choosing a value of $\eta$ appears to require knowledge of $\Delta^{\circ}$ and $T$ in advance. However, the so-called ``generic doubling trick'' allows the bidder to learn these values adaptively at the price of a constant factor~\cite{HazKal10}. In the partial information case, this change also requires replacing $\Delta^\circ$, the width of an interval containing an optimal bid, by $\Delta$, with width of the narrowest interval.

We initialize two bounds, $B_T=1$ and $B_\Delta=1$, and run \exptree{} with parameter $\eta = (1/2)\sqrt{B_\Delta/B_T}\wedge\frac 1 2$ until either $t \leq B_T$ or $\log 1/\Delta \leq B_\Delta$ fails to hold. When one of these bounds is breached, we double the bound and restart the algorithm, maintaining the partition $\cL_t$ but setting $w_\ell = 1$ for all $\ell \in \cL_t$. This modified strategy yields the following theorem.
\begin{theorem}
\label{TH:adv02}
The strategy \exptree{} run with the above doubling procedure yields an expected regret bound
\begin{equation*}
R_T \leq 48 \sqrt{2 T \log (1/\Delta)}\,.
\end{equation*}
\end{theorem}
\begin{proof}
Divide the algorithm into stages on which $B_T$ and $B_\Delta$ are constant, and denote by $B_T^{\star}$ and $B_\Delta^{\star}$ the values of $B_T$ and $B_\Delta$ when the algorithm terminates. The proof of Theorem~\ref{TH:exptree} implies that the expected regret incurred during any given stage is at most
\begin{equation*}
4\eta T + \frac{\log(1/\Delta)}{\eta} \leq 4\eta B_T + \frac{B_\Delta}{\eta} \leq 4 \sqrt{T \log(1/\Delta)}\,.
\end{equation*}
It remains to sum these regrets over each stage, since the actual expected regret (which requires a fixed bid across all stages) can only be smaller. 

Suppose that the algorithm lasted a total of $\ell+ m+1$ stages, $\ell$ of which were ended because the bound $t \leq B_T$ was violated and $m$ of which were ended because the bound $\log(1/\Delta) \leq B_\Delta$ was violated. The total regret across all $\ell + m + 1$ stages is bounded by
\begin{equation*}
\sum_{i = 0}^{\ell} \sum_{j = 0}^{m} 4\sqrt{2^{i} \cdot 2^{j}} = \frac{4}{(\sqrt 2 - 1)^2} (2^{(i+1)/2} - 1)(2^{(j+1)/2}-1) \leq 48 \sqrt{B_T^* B_\Delta^\star}\,.
\end{equation*}
Moreover, when the algorithm terminates, we have the bounds $B_T^{\star} \leq T$ and $B_\Delta^{\star} \leq 2 \log(1/\Delta)$. The result follows.
\end{proof}

\subsection{Adaptive adversaries}

Theorem~\ref{TH:exptree} establishes an upper bound on the pseudo-regret against any adversary. Moreover, when the adversary is oblivious, the same bound holds for the expected regret. When the adversary is adaptive, however, achieving a bound on the expected regret requires a slightly modified algorithm, Algorithm~\ref{alg:exptreep}.

Algorithm~\ref{alg:exptreep} differs from Algorithm~\ref{alg:exptree} chiefly in the method of calculating the estimated gain in~\eqref{EQ:estgain1}. In place of $\hat g(\ell, t)$, \exptreep{} employs a \emph{biased} estimate $\tilde g(\ell, t)$ defined by
\begin{equation}
\label{EQ:estgain2}
\tilde g(\ell, t)=\frac{v_t \1\{b_t > m_t\}+\beta}{\p_{ \cB_t} ( b> m_t)} \1\{m_t\preceq \ell\} + \frac{m_t\1\{b_t \le m_t\}+\beta}{1-\p_{ \cB_t} (b_t > m_t)}\1\{m_t\succeq \ell \}\,.
\end{equation}

The following theorem holds.

\begin{theorem}
\label{TH:exptreep}
Let $v_1, \dots, v_T \in [0,1]$ and $m_1, \dots, m_T \in [0,1]$ be arbitrary sequences.  Let $\ell^{\circ} \in \cL_T$ denote the narrowest interval in the finest partition $\cL_T$ and let $\Delta^\circ=|\ell^\circ|$ denote its width. The strategy \exptreep{} run with parameters
\begin{equation*}
\eta = \sqrt{\frac{\log(1/\Delta^\circ)}{8T}}\wedge\frac 1 8, \quad \gamma = 2 \eta, \quad\text{and } \beta = \sqrt{\frac{\log T}{2 T}}
\end{equation*}
yields
$$
R_T \leq 2\sqrt{8T \log(1/\Delta^\circ)} + 3\sqrt{2T \log T}\log(1/\delta)\,,
$$
with probability at least $1-\delta$. Moreover,
$$
\E[R_T]\le 2\sqrt{8T \log(1/\Delta^\circ)} + 3\sqrt{2T \log T}\,.
$$
\end{theorem}

Denote by $G(b) = \sum_{t=1}^T G(b, t)$ and $\tilde G(b) = \sum_{t=1}^T \tilde g(b, t)$ the cumulative true and estimated gains for a bet $b$. Before proving Theorem~\ref{TH:exptreep}, we establish the following lemma, which shows that $\tilde G$ can be viewed as an upper bound on $G$. 
\begin{lemma}\label{LEM1}
With probability at least $1- \delta$ the bound
$
G(b, T) \leq \tilde G(b, T)  + \frac{\log{T \delta^{-1}}}{\beta}
$
holds for all $b \in [0, 1]$.
\end{lemma}
\begin{proof}
Denote by $\E_t$ expectation with respect to the random choice of $b_t$, conditioned on the outcomes of rounds $1, \dots, t-1$. Fix $b \in [0, 1]$ and define $d(b, t) = g(b, t) - \tilde g(b, t)$. Note that
$$
\E_t[d(b, t)]=-\frac{\beta\1\{b>m_t\}}{\p_{ \cB_t} ( b_t> m_t)} - \frac{\beta\1\{b \leq m_t\}}{1-\p_{ \cB_t} ( b_t> m_t)}
$$
so that
\begin{equation}
\label{EQ:diffexp}
\bar d(b,t):=d(b,t)-\E_t[d(b, t)] = v_t\1\{b>m_t\}\big(1-\frac{\1\{b_t>m_t\}}{\p_{ \cB_t} ( b_t> m_t)}\big)  +m_t\1\{b<m_t\}\big(1-\frac{\1\{b_t<m_t\}}{\p_{ \cB_t} ( b_t< m_t)}\big) 
\end{equation}
This immediately immediately yields $\bar d(b,t)\le g(b,t)\leq 1$.
Since $\beta \leq 1$, and $e^x \leq 1 + x + x^2$ for $x \leq 1$, we have
\begin{align*}
\E_t\big[e^{\beta d (b, t)}\big] & = e^{ \beta \E_t[d(b, t)]} \E_t\big[e^{\beta \bar d(b, t)}\big] \leq e^{ \beta \E_t[d(b, t)]} \big(1 + \beta^2 \E_t[\bar d\,^2(b, t)]\big).
\end{align*}
It follows from (\ref{EQ:diffexp}) that
\begin{align*}
\E_t[\bar d\,^2(b, t)] & \leq \frac{v_t^2}{\p_{ \cB_t}(b_t>m_t)}\1\{b>m_t\} + \frac{m_t^2}{1- \p_{ \cB_t}(b_t>m_t)}\1\{b\leq m_t\} \\
& \leq \frac{1}{\p_{ \cB_t}(b_t>m_t)}\1\{b>m_t\} + \frac{1}{1- \p_{ \cB_t}(b_t>m_t)}\1\{b\leq m_t\} \\
& = - \frac{1}{\beta} \E_t[d(b, t)].
\end{align*}
Combining this with the preceding inequality and the fact that $\beta \E_t[d(b, t)]\le 1$ yields
\begin{equation*}
\E_t\big[e^{\beta d (b, t)}\big] \leq e^{ \beta \E_t[d(b, t)]}\big(1 - \beta\E_t[d(b, t)]\big) \leq 1.
\end{equation*}
Let $Z_t = \exp(\beta  d(b, t))$. Then 
\begin{equation*}
\E\big[e^{\beta( G(b) - \tilde G(b))}\big] \leq  \E\big[\exp\big(\beta \sum_{t=1}^T  d(b, t)\big)\big] =  \E\big[\prod_{t=1}^T Z_t\big] \leq 1,
\end{equation*}
where the last step follows by conditioning on each stage in turn and applying the above bound.

To obtain a uniform bound, we note that the function $b \mapsto G(b)-\tilde G(b)$ takes at most $T$ random values $G_1, \ldots, G_T$ as $b$ varies across $[0,1]$. Moreover, we have just proved that $\max_j \E[\exp(\beta G_j)]\le 1$. 
Hence
\begin{align*}
\E\big[\exp\big(\beta \big[ \max_{b\in[0,1]} G(b) - \tilde G(b)\big]\big)\big]  =\E\big[\exp\big(\beta  \max_{j\in[T]} G_j\big)\big]  \leq \sum_{j=1}^T \E\big[ e^{\beta G_j}\big]  \leq T.
\end{align*}

Applying the Markov bound yields the claim.
\end{proof}
We are now in a position to prove Theorem~\ref{TH:exptreep}.

\begin{proof}[Proof of Theorem~\ref{TH:exptreep}]
We proceed as in the proof of Theorem~\ref{TH:exptree}. Note that the choice of $\eta$ guarantees that $\cB_t$ is a valid probability distribution.

As above, define $W_t=\sum_{\kappa \in \cL_t}|\kappa| w_{\kappa,t}$. We have
\begin{align*}
\log \frac{W_{t+1}}{W_{t}} & = \log \sum_{\ell \in \cL_{t+1}} \frac{|\ell|w_{\ell, t} \exp(\eta \tilde g(\ell, t))}{W_t}  = \log \sum_{\ell \in \cL_{t+1}} p_{\ell, t} \exp(\eta\tilde g(\ell, t)).
\end{align*}
Since $\eta \tilde g(\ell, t) \leq \eta \frac{1 + \beta}{\gamma}\leq 1$, the inequality $e^x \leq 1 + x + x^2$ for $x \leq 1$ implies
\begin{equation*}
\log \frac{W_{t+1}}{W_{t}}  \leq \log \sum_{\ell \in \cL_{t+1}} p_{\ell, t} (1 + \eta \tilde g(\ell, t) + \eta^2 \tilde g(\ell, t)^2)  = \log \big( 1+ \eta \sum_{\ell \in \cL_{t+1}} p_{\ell, t} \tilde g(\ell, t) + \eta^2 \sum_{\ell \in \cL_{t+1}} p_{\ell, t} \tilde g(\ell, t)^2\big).
\end{equation*}

By the same reasoning as in the proof of Theorem~\ref{TH:exptree}, we have
\begin{equation*}
\sum_{\ell \in \cL_{t+1}} p_{\ell, t} \tilde g(\ell, t) = \frac{1}{1-2\gamma}\sum_{\ell \in \cL_{t+1}} \p_{\cB_t}(b_t \in \ell) \tilde g(\ell, t) \leq \frac{1}{1-2\gamma} (g(b_t, t) + 2 \beta)
\end{equation*}
and similarly
\begin{equation*}
\sum_{\ell \in \cL_{t+1}} p_{\ell, t} \tilde g(\ell, t)^2 = \frac{1}{1-2\gamma}\sum_{\ell \in \cL_{t+1}} \p_{\cB_t}(b_t \in \ell) \tilde g(\ell, t)^2.
\end{equation*}
To compute this last quantity, note that (\ref{EQ:estgain2}) implies
\begin{align*}
\sum_{\ell \in \cL_{t+1}} \p_{\cB_t}(b_t \in \ell) \tilde g(\ell, t)^2 &= \sum_{\stackrel{\ell \in \cL_{t+1}}{ m_t \preceq \ell} }\frac{\p_{\cB_t}(b_t \in \ell)}{\p_{ \cB_t} ( b_t> m_t)} (v_t\1\{b_t>m_t\} + \beta) \tilde g(\ell, t)\\&\phantom{=}+\sum_{\stackrel{\ell \in \cL_{t+1}}{m_t \succeq \ell}} \frac{\p_{\cB_t}(b_t \in \ell)}{1-\p_{ \cB_t} ( b_t> m_t)} (m_t\1\{b_t\le m_t\} + \beta)\tilde g(\ell, t) \\
& \leq g(b_t, t)\tilde g(b_t, t) + \beta \big(\frac{v_t\1\{b_t>m_t\} + \beta}{\p_{\cB_t}(b_t \in \ell)} + \frac{m_t\1\{b_t\leq m_t\} + \beta}{1- \p_{\cB_t}(b_t \in \ell)}\big) \\
& = g(b_t, t)\tilde g(b_t, t) + \beta(\tilde g(1, t) + \tilde g(0, t))\\
& \leq (1+ \beta)(\tilde g(1, t) + \tilde g(0, t))\,. 
\end{align*}
where in the last inequality, we used the fact $g(b_t,t)\le 1$ and $\tilde g(b_t, t) \le \tilde g(1, t) + \tilde g(0, t)$.
Combining the above bounds yields
\begin{align*}
\log \frac{W_{t+1}}{W_{t}} \leq \frac{\eta}{1- 2\gamma} (g(b_t, t) + 2\beta) + \frac{\eta^2}{1 - 2 \gamma}  (1+\beta)(\tilde g(1, t) + \tilde g(0, t)).
\end{align*}

Defining $\bar G=\sum_{t=1}^T g(b_t, t)$ and summing on $T$ yields
\begin{align*}
\log \frac{W_T}{W_0} & \leq \frac{\eta}{1-2\gamma}  \bar G + \frac{2 T \eta \beta}{1 - 2\gamma} + \frac{\eta^2}{1 - 2 \gamma}(1+ \beta) (\tilde G(1, T) + \tilde G(0, t)) \\
& \leq \frac{\eta}{1-2\gamma}  \bar G + \frac{2 T \eta \beta}{1 - 2\gamma} + \frac{2\eta^2}{1 - 2 \gamma}(1+ \beta) \max_{b \in [0, 1]} \tilde G(b).
\end{align*}

We bound $W_T$ by writing
\begin{equation*}
\log W_T \geq \log \Delta^{\circ} + \eta \max_{b \in [0, 1]} \tilde G(b).
\end{equation*}

Rearranging yields
\begin{equation*}
(1- 2\gamma - 2 \eta(1+\beta)) \max_{b \in [0, 1]} \tilde G(b) - \bar G \leq 2 T \beta +
\frac{(1-2 \gamma)\log(1/\Delta^{\circ})}{\eta} \leq 2 T \beta +
\frac{\log(1/\Delta^{\circ})}{\eta}.
\end{equation*}

Applying Lemma~\ref{LEM1}, with probability $1-\delta$ we have
\begin{equation*}
\max_{b \in [0, 1]}  G(b) \leq \max_{b \in [0, 1]} \tilde G(b) + \frac{\log(T\delta^{-1})}{\beta}
\end{equation*}
which implies
\begin{equation*}
(1- 2\gamma - 2 \eta(1+\beta)) \max_{b \in [0, 1]} G(b, T) - \bar G \leq T \beta + \frac{\log(T\delta^{-1})}{\beta} + \frac{\log(1/\Delta^{\circ})}{\eta}
\end{equation*}
since $2\gamma + 2\eta(2+\beta) \leq 8 \eta \leq 1$. We obtain
\begin{align*}
\max_{b \in [0,1]} G(b) - \bar G & \leq 2 T \beta + \frac{\log(T\delta^{-1})}{\beta} + \frac{\log(1/\Delta^{\circ})}{\eta} + (2 \gamma + 2\eta(1+\beta))T \\ 
& \leq  2 T\beta + 8 \eta T + \frac{\log(T\delta^{-1})}{\beta}+ \frac{\log(1/\Delta^{\circ})}{\eta}.
\end{align*}

Plugging in the given parameters then yields the claim.

The bound in expectation follows upon integrating the first result.
\end{proof}

\subsection{Lower bound}

The dependence on $\Delta^\circ$ in Theorem~\ref{TH:exptreep} is unfortunate, since the resulting bounds become vacuous when $\Delta^\circ$ is exponentially small. However, it turns out that this dependence is unavoidable. We prove in this section a lower bound on the pseudo-regret $\bar R_T$. Since $\bar R_T \leq \E R_T$, this bound will also hold for expected regret.

We begin with a lemma establishing that the rate $\sqrt{T}$ is optimal, using standard information theoretic techniques for lower bounds (see, e.g., \cite{Tsy09}).

\begin{lemma}\label{lem:lb1}
Fix an $m \in [1/4, 3/4]$. There exist a pair of adversaries $U$ and $L$ such that $m_t = m$ for all $t$ and the sequence $v_1, \dots, v_T$ is \iid conditional on the choice of adversary and such that
\begin{equation*}
\max_{A \in \{U, L\}} \max_{b \in [0, 1]} \E_A\big[\sum_{t=1}^T g(b, t) - \sum_{t=1}^T g(b_t, t)\big] \geq \frac{1}{32} \sqrt{T}.
\end{equation*}
Moreover, under adversary $U$ any bet $b > m$ is optimal, and under adversary $L$ any bet $b < m$ is optimal.
\end{lemma}
\begin{proof}
We first consider deterministic strategies. Fix an $\varepsilon$. Denote by $U$ the adversary under which $v_t \sim \bern(m + \varepsilon)$ and by $L$ the adversary under which $v_t \sim \bern(m - \varepsilon)$. 

Given a sequence of bids $b_1, \dots, b_T$, let $T_{\tinyminus}$ and $T_{\tinyplus}$ be the number of times $t$ for which $b_t < m$ and $b_t > m$, respectively.
 Denoting the regret after $T$ rounds by $R_T$, it is easy to show that

\begin{align*}
\E_U[R_T] & \geq \varepsilon \E_U[T_{\tinyminus}] \\
\E_L[R_T] & \geq \varepsilon \E_L[T_{\tinyminus}].
\end{align*}

Write $\p_U$ and $\p_L$ for the law of $T_{\tinyminus}$ under adversary $U$ and $L$, respectively, and denote by $\p_{\text{av}}$ the distribution of $T_{\tinyminus}$ when $v_t \sim \bern (m)$. Then Pinsker's inequality implies
\begin{align*}
\E_U[T_{\tinyminus}]  \geq \E_{\text{av}}(T_{\tinyminus}) - T\sqrt{\KL(\p_U, \p_{\text{av}})/2}\,, \quad \E_L[T_{\tinyplus}] \geq \E_{\text{av}}(T_{\tinyplus}) - T\sqrt{ \KL(\p_L, \p_{\text{av}})/2}\,.
\end{align*}

By the data processing inequality, 
\begin{equation*}
\KL(\p_U, \p_{\text{av}}) \leq T \cdot \KL(\bern(m + \varepsilon), \bern(m)) \leq T\frac{\varepsilon^2}{m(1-m)} \leq 8T\varepsilon^2,
\end{equation*}
and likewise for $\KL(\p_L, \p_{\text{av}})$. We therefore obtain
\begin{equation*}
\frac{1}{2}\big( \E_U[R_T] + \E_L[R_T]\big) \geq \varepsilon\big(\frac{T}{2} - 2T \varepsilon \sqrt T\big).
\end{equation*}
Setting $\varepsilon = \frac{1}{8 \sqrt T}$ and bounding the average by a maximum yields
\begin{equation*}
\max_{A \in \{U, L\}} \max_{b \in [0, 1]} \E_A\big[\sum_{t=1}^T g(b, t) - \sum_{t=1}^T g(b_t, t)\big] \geq \frac{\sqrt{T}}{32}
\end{equation*}
for any deterministic strategy. The claim follows for general strategies by averaging over the bidder's internal randomness and applying Fubini's theorem.
\end{proof}

We are now in a position to prove a tight minimax lower bound.

\begin{theorem}
For any strategy and any value of $\Delta^\circ \in (0, 1/4)$, there exists sequences $v_1, \dots, v_T \in [0,1]$ and $m_1, \dots, m_T \in [0,1]$ such that $\Delta^\circ$ is the smallest positive gap between the adversary's bids and
\begin{equation*}
\max_{b \in [0, 1]} \E \sum_{t=1}^T g(b, t) - \E \sum_{t=1}^T g(b_t, t) \geq \frac{1}{32}\sqrt{T \lfloor \log_2 (1/2\Delta^\circ)\rfloor}.
\end{equation*}
\end{theorem}
\begin{proof}
We can assume without loss of generality that $\Delta$ is a power of $2$, since this can change the regret by at most a constant.

 Set $n = \log_2 (1/2\Delta)$. We divide the game into $n$ stages of $\frac T n$ rounds each and will show that any bidder incurs regret of at least $\frac{1}{32}\sqrt{\frac T n}$ during each stage by repeatedly applying Lemma~\ref{lem:lb1}.

During the first stage, apply Lemma~\ref{lem:lb1} with $m  = 1/2$. One of the two adversaries will incur regret in expectation of at least $\frac{1}{32}\sqrt{\frac T n}$. If that adversary is $U$, the next stage will use Lemma~\ref{lem:lb1} with $m=5/8$; if it is $L$, then the next stage will use $m = 3/8$.

In general, for the $i$th stage we will apply Lemma~\ref{lem:lb1} with $m = 1/4 + c_i2^{-i-1}$ for some $c_i$. If the $U$ adversary has higher regret in expectation at that stage, then $c_{i+1} = 2c_i + 1$; otherwise $c_{i+1} = 2c_i - 1$. Note that during the $i$th stage, the smallest gap between two of the adversary's bids is $2^{-i-1}$. 

The structure of the optimum bids for the adversaries $U$ and $L$ guarantees that during each stage, there is an interval within which a fixed bid would be optimal for all previous stages. So after $n$ stages there is a fixed bid that is optimal for all $n$ adversaries. Therefore the regret across the $n$ stages is equal to the sum of the regrets for each stage, and we obtain
\begin{equation*}
\max_{b \in [0, 1]} \E \sum_{t=1}^T g(b, t) - \E \sum_{t=1}^T g(b_t, t) \geq n \frac{1}{32} \sqrt{\frac{T}{n}} = \frac{1}{32} \sqrt{T n} = \frac{1}{32}\sqrt{T \lfloor \log_2 (1/2\Delta)\rfloor},
\end{equation*}
as desired.
\end{proof}

\section{Conclusion and open questions}

Building on established strategies for the bandit problem, we propose a  first set of strategies tailored to online learning in repeated auctions. Depending on the model, stochastic or adversarial, we obtain  several regret bounds ranging from $O(\log T)$ to $O(\sqrt{T})$ and exhibit a reasonable family of models where regret bounds $\tilde O(T^{\beta/2})$ are achievable for all $\beta \in (0,1)$.

In both setups, several questions are beyond the scope of this paper and are left open.
\begin{enumerate}
\item What is the effect of covariates on this problem? In practice, potentially relevant information about the value of the good is available before bidding \cite{MohMed14} and incorporating such covariates can allow for a better model. This question falls into the realm of \emph{contextual bandits} that has been studied both in the stochastic and the adversarial framework \cite{WanKulPoo03, KakShaTew08,  BubCes12, PerRig13, Sli14}.
\item In the adversarial case, our benchmark is the best fixed bid in hindsight. While this is rather standard in the online learning literature, recent developments have allowed for more complicated benchmarks, namely sophisticated but fixed strategies \cite{HanRakSri13}. Such developments are available only for the full information case, however.
\item Our results indicate that when facing well behaved bidders, better regret bounds are achievable in the stochastic case. Similar results are of interest in the adversarial case too \cite{HazKal10, RakSri13, FosRakSri15}. Here too, unfortunately, existing results are limited to the full information case.
\item The proof of Theorem \ref{TH:exptreep} involves a union bound which leads to a  $O(\sqrt{T\log(T)}\log(1/\delta))$ regret upper bound. The result is a gap of order  $\sqrt{\log(T)}$ between the upper and lower bound. Is this term really present?
\end{enumerate} 
\newpage

\newpage
\bibliographystyle{amsalpha}
\bibliography{rigolletmain}

\newcommand{\etalchar}[1]{$^{#1}$}
\providecommand{\bysame}{\leavevmode\hbox to3em{\hrulefill}\thinspace}
\providecommand{\MR}{\relax\ifhmode\unskip\space\fi MR }
\providecommand{\MRhref}[2]{%
  \href{http://www.ams.org/mathscinet-getitem?mr=#1}{#2}
}
\providecommand{\href}[2]{#2}
\begin{thebibliography}{ACBFS03}

\bibitem[ACBF02]{AueCesFis02}
Peter Auer, Nicol\`{o} Cesa-Bianchi, and Paul Fischer, \emph{Finite-time
  analysis of the multiarmed bandit problem}, Mach. Learn. \textbf{47} (2002),
  no.~2-3, 235--256.

\bibitem[ACBFS03]{AueCesFre01}
Peter Auer, Nicol{\`o} Cesa-Bianchi, Yoav Freund, and Robert~E. Schapire,
  \emph{The nonstochastic multiarmed bandit problem}, SIAM J. Comput.
  \textbf{32} (2002/03), no.~1, 48--77 (electronic). \MR{1954855 (2003k:91031)}

\bibitem[ACD{\etalchar{+}}15]{AmiCumDwo15}
Kareem Amin, Rachel Cummings, Lili Dworkin, Michael Kearns, and Aaron Roth,
  \emph{Online learning and profit maximization from revealed preferences},
  Proceedings of the Twenty-Ninth {AAAI} Conference on Artificial Intelligence,
  January 25-30, 2015, Austin, Texas, {USA.}, 2015, pp.~770--776.

\bibitem[ARS14]{AmiRosSye14}
Kareem Amin, Afshin Rostamizadeh, and Umar Syed, \emph{Repeated contextual
  auctions with strategic buyers}, Advances in Neural Information Processing
  Systems 27 (Z.~Ghahramani, M.~Welling, C.~Cortes, N.D. Lawrence, and K.Q.
  Weinberger, eds.), Curran Associates, Inc., 2014, pp.~622--630.

\bibitem[BCB12]{BubCes12}
S{\'e}bastien Bubeck and Nicol{\`o} Cesa-Bianchi, \emph{Regret analysis of
  stochastic and nonstochastic multi-armed bandit problems}, Foundations and
  Trends in Machine Learning \textbf{5} (2012), no.~1, 1--122.

\bibitem[BKS10]{BabKleSli10}
Moshe Babaioff, Robert~D. Kleinberg, and Aleksandrs Slivkins, \emph{Truthful
  mechanisms with implicit payment computation}, Proceedings of the 11th ACM
  Conference on Electronic Commerce (New York, NY, USA), EC '10, ACM, 2010,
  pp.~43--52.
  
  \bibitem[BFP{\etalchar{+}}14]{BarFosPal14}
G{\'a}bor Bart{\'o}k, Dean~P. Foster, D{\'a}vid P{\'a}l, Alexander Rakhlin, and
  Csaba Szepesv{\'a}ri, \emph{Partial monitoring---classification, regret
  bounds, and algorithms}, Math. Oper. Res. \textbf{39} (2014), no.~4,
  967--997. \MR{3279754}

\bibitem[BMM15]{BluManMor15}
Avrim Blum, Yishay Mansour, and Jamie Morgenstern, \emph{Learning valuation
  distributions from partial observation}, Proceedings of the Twenty-Ninth
  {AAAI} Conference on Artificial Intelligence, January 25-30, 2015, Austin,
  Texas, {USA.}, 2015, pp.~798--804.

\bibitem[BPR13]{BubPerRig13}
S{{\'e}}bastien Bubeck, Vianney Perchet, and Philippe Rigollet, \emph{Bounded
  regret in stochastic multi-armed bandits}, COLT 2013 - The 26th Conference on
  Learning Theory, Princeton, NJ, June 12-14, 2013 (Shai Shalev-Shwartz and
  Ingo Steinwart, eds.), JMLR W\&CP, vol.~30, 2013, pp.~122--134.

\bibitem[BSS09]{BabShaSli09}
Moshe Babaioff, Yogeshwer Sharma, and Aleksandrs Slivkins, \emph{Characterizing
  truthful multi-armed bandit mechanisms: Extended abstract}, Proceedings of
  the 10th ACM Conference on Electronic Commerce (New York, NY, USA), EC '09,
  ACM, 2009, pp.~79--88.

\bibitem[CBGM13]{CesGenMan13}
Nicol\`{o} Cesa-Bianchi, Claudio Gentile, and Yishay Mansour, \emph{Regret
  minimization for reserve prices in second-price auctions}, Proceedings of the
  Twenty-Fourth Annual ACM-SIAM Symposium on Discrete Algorithms, SODA '13,
  SIAM, 2013, pp.~1190--1204.

\bibitem[CBL06]{CesLug06}
Nicol{{\`o}} Cesa-Bianchi and G{{\'a}}bor Lugosi, \emph{Prediction, learning,
  and games}, Cambridge University Press, Cambridge, 2006. \MR{2409394
  (2009g:91006)}

\bibitem[CHN14]{ChaHarNek14}
Shuchi Chawla, Jason Hartline, and Denis Nekipelov, \emph{Mechanism design for
  data science}, Proceedings of the Fifteenth ACM Conference on Economics and
  Computation (New York, NY, USA), EC '14, ACM, 2014, pp.~711--712.

\bibitem[CM88]{CreMcL88}
Jacques Cr{\'e}mer and Richard~P. McLean, \emph{Full extraction of the surplus
  in bayesian and dominant strategy auctions}, Econometrica \textbf{56} (1988),
  no.~6, pp. 1247--1257 (English).

\bibitem[CR14]{ColRou14}
Richard Cole and Tim Roughgarden, \emph{The sample complexity of revenue
  maximization}, Proceedings of the 46th Annual ACM Symposium on Theory of
  Computing (New York, NY, USA), STOC '14, ACM, 2014, pp.~243--252.

\bibitem[DK09]{DevKak09}
Nikhil~R. Devanur and Sham~M. Kakade, \emph{The price of truthfulness for
  pay-per-click auctions}, Proceedings of the 10th ACM Conference on Electronic
  Commerce (New York, NY, USA), EC '09, ACM, 2009, pp.~99--106.

\bibitem[DRY15]{DhaRouYan15}
Peerapong Dhangwatnotai, Tim Roughgarden, and Qiqi Yan, \emph{Revenue
  maximization with a single sample}, Games Econom. Behav. \textbf{91} (2015),
  318--333. \MR{3353730}

\bibitem[FHH13]{FuHarHoy13}
Hu~Fu, Jason Hartline, and Darrell Hoy, \emph{Prior-independent auctions for
  risk-averse agents}, Proceedings of the Fourteenth ACM Conference on
  Electronic Commerce (New York, NY, USA), EC '13, ACM, 2013, pp.~471--488.

\bibitem[FHHK14]{FuHagHar14}
Hu~Fu, Nima Haghpanah, Jason Hartline, and Robert Kleinberg, \emph{Optimal
  auctions for correlated buyers with sampling}, Proceedings of the Fifteenth
  ACM Conference on Economics and Computation (New York, NY, USA), EC '14, ACM,
  2014, pp.~23--36.

\bibitem[FRS15]{FosRakSri15}
Dylan Foster, Alexander Rakhlin, and Karthik Sridharan, \emph{Adaptive online
  learning}, NIPS, 2015.

\bibitem[HK10]{HazKal10}
Elad Hazan and Satyen Kale, \emph{Extracting certainty from uncertainty: regret
  bounded by variation in costs}, Mach. Learn. \textbf{80} (2010), no.~2-3,
  165--188. \MR{3108164}

\bibitem[HR09]{HarRou09}
Jason~D. Hartline and Tim Roughgarden, \emph{Simple versus optimal mechanisms},
  SIGecom Exch. \textbf{8} (2009), no.~1, 5:1--5:3.

\bibitem[HRS13]{HanRakSri13}
Wei Han, Alexander Rakhlin, and Karthik Sridharan, \emph{Competing with
  strategies}, COLT 2013 - The 26th Conference on Learning Theory, Princeton,
  NJ, June 12-14, 2013 (Shai Shalev-Shwartz and Ingo Steinwart, eds.), JMLR
  W\&CP, vol.~30, 2013, pp.~966--992.

\bibitem[KN14]{KanNaz14}
Yash Kanoria and Hamid Nazerzadeh, \emph{Dynamic reserve prices for repeated
  auctions: Learning from bids}, Web and Internet Economics (Tie-Yan Liu,
  Qi~Qi, and Yinyu Ye, eds.), Lecture Notes in Computer Science, vol. 8877,
  Springer International Publishing, 2014, pp.~232--232 (English).

\bibitem[KSST08]{KakShaTew08}
Sham~M. Kakade, Shai Shalev-Shwartz, and Ambuj Tewari, \emph{Efficient bandit
  algorithms for online multiclass prediction}, ICML (William~W. Cohen, Andrew
  McCallum, and Sam~T. Roweis, eds.), ACM International Conference Proceeding
  Series, vol. 307, ACM, 2008, pp.~440--447.

\bibitem[LR85]{LaiRob85}
T.~L. Lai and H.~Robbins, \emph{Asymptotically efficient adaptive allocation
  rules}, Advances in Applied Mathematics \textbf{6} (1985), 4--22.

\bibitem[McA11]{McA11}
R.Preston McAfee, \emph{The design of advertising exchanges}, Review of
  Industrial Organization \textbf{39} (2011), no.~3, 169--185 (English).

\bibitem[MM14]{MohMed14}
Mehryar Mohri and Andres~Mu{\~{n}}oz Medina, \emph{Learning theory and
  algorithms for revenue optimization in second price auctions with reserve},
  Proceedings of the 31th International Conference on Machine Learning, {ICML}
  2014, Beijing, China, 21-26 June 2014, 2014, pp.~262--270.

\bibitem[MT99]{MamTsy99}
E.~Mammen and A.~B. Tsybakov, \emph{Smooth discrimination analysis}, Ann.
  Statist. \textbf{27} (1999), no.~6, 1808--1829. \MR{MR1765618 (2001i:62074)}

\bibitem[Mut09]{Mut09}
S.~Muthukrishnan, \emph{Ad exchanges: Research issues}, Internet and Network
  Economics (Stefano Leonardi, ed.), Lecture Notes in Computer Science, vol.
  5929, Springer Berlin Heidelberg, 2009, pp.~1--12 (English).

\bibitem[Mye81]{Mye81}
Roger~B. Myerson, \emph{Optimal auction design}, Math. Oper. Res. \textbf{6}
  (1981), no.~1, 58--73. \MR{618964 (82m:90191)}

\bibitem[OS11]{OstSch11}
Michael Ostrovsky and Michael Schwarz, \emph{Reserve prices in internet
  advertising auctions: A field experiment}, Proceedings of the 12th ACM
  Conference on Electronic Commerce (New York, NY, USA), EC '11, ACM, 2011,
  pp.~59--60.

\bibitem[PR13]{PerRig13}
Vianney Perchet and Philippe Rigollet, \emph{The multi-armed bandit problem
  with covariates}, Ann. Statist. \textbf{41} (2013), no.~2, 693--721.

\bibitem[RS81]{RilSam81}
John Riley and William~F Samuelson, \emph{Optimal auctions}, American Economic
  Review \textbf{71} (1981), no.~3, 381--92.

\bibitem[RS13]{RakSri13}
Alexander Rakhlin and Karthik Sridharan, \emph{Online learning with predictable
  sequences}, COLT 2013 - The 26th Conference on Learning Theory, Princeton,
  NJ, June 12-14, 2013 (Shai Shalev-Shwartz and Ingo Steinwart, eds.), JMLR
  W\&CP, vol.~30, 2013, pp.~993--1019.

\bibitem[RTCY12]{RouTalYan12}
Tim Roughgarden, Inbal Talgam-Cohen, and Qiqi Yan, \emph{Supply-limiting
  mechanisms}, Proceedings of the 13th ACM Conference on Electronic Commerce
  (New York, NY, USA), EC '12, ACM, 2012, pp.~844--861.

\bibitem[Sli14]{Sli14}
Aleksandrs Slivkins, \emph{Contextual bandits with similarity information}, J.
  Mach. Learn. Res. \textbf{15} (2014), no.~1, 2533--2568.

\bibitem[Tsy06]{Tsy06}
Alexandre Tsybakov, \emph{Statistique appliqu\'ee}, Lecture Notes, 2006.

\bibitem[Tsy09]{Tsy09}
Alexandre~B. Tsybakov, \emph{Introduction to nonparametric estimation},
  Springer Series in Statistics, Springer, New York, 2009, Revised and extended
  from the 2004 French original, Translated by Vladimir Zaiats. \MR{2724359
  (2011g:62006)}

\bibitem[Wil69]{Wil69}
Robert~B. Wilson, \emph{Competitive bidding with disparate information},
  Management Science \textbf{15} (1969), no.~7, 446--448.

\bibitem[Wil87]{Wil87}
Robert Wilson, \emph{Game-theoretic analyses of trading processes}, Advances in
  Economic Theory (Truman~Fassett Bewley, ed.), Cambridge University Press,
  1987, Cambridge Books Online, pp.~33--70.

\bibitem[WKP03]{WanKulPoo03}
Chih-Chun Wang, S.R. Kulkarni, and H.V. Poor, \emph{Bandit problems with
  arbitrary side observations}, Decision and Control, 2003. Proceedings. 42nd
  IEEE Conference on, vol.~3, Dec 2003, pp.~2948--2953 Vol.3.

\end{thebibliography}
\end{document}